\pdfoutput=1
\documentclass{article}
\usepackage{hyperref}       
\usepackage{url}            
\usepackage{graphicx}
\usepackage{amsmath,amssymb,amsthm,amsfonts}
\usepackage{hyperref}
\usepackage{enumitem}
\usepackage{wrapfig}
\usepackage[margin=1in]{geometry}
\usepackage{algorithm}
\usepackage{algorithmic}
\usepackage{subfig}

\theoremstyle{plain}
\newtheorem{theorem}{Theorem}[section]
\newtheorem{lemma}[theorem]{Lemma}
\newtheorem{claim}{Claim}

\newtheorem{corollary}[theorem]{Corollary}

\theoremstyle{definition}
\newtheorem{definition}{Definition}[section]
\newtheorem{remark}{Remark}
\newcommand{\ignore}[1]{}

\DeclareMathOperator*{\E}{\mathbb{E}}

\DeclareMathOperator*{\argmin}{argmin}

\DeclareMathOperator{\nnz}{nnz}
\DeclareMathOperator{\diag}{diag}
\DeclareMathOperator{\poly}{poly}
\DeclareMathOperator{\polylog}{polylog}
\renewcommand{\epsilon}{\varepsilon}
\newcommand{\eps}{\varepsilon}

\newcommand{\RR}{\mathbb{R}}

\title{Nearly Linear Row Sampling Algorithm for Quantile Regression}

\author{
Yi Li\\
{\small School of Physical and Mathematical Sciences}\\
{\small Nanyang Technological University, Singapore}\\
{\small \texttt{yili@ntu.edu.sg}}
\and
Ruosong Wang\\
{\small Department of Computer Science}\\
{\small Carnegie Mellon University, USA}\\
{\small \texttt{ruosongw@andrew.cmu.edu}}
\and
Lin Yang\\
{\small Department of Electrical and Computer Engineering}\\
{\small University of California Los Angeles, USA}\\
{\small \texttt{linyang@ee.ucla.edu}}
\and
Hanrui Zhang\\
{\small Department of Computer Science}\\
{\small Duke University, USA}\\
{\small \texttt{hrzhang@cs.duke.edu}}
}

\date{}

\begin{document}
\maketitle

\begin{abstract}
We give a row sampling algorithm for the quantile loss function with sample complexity nearly linear in the dimensionality of the data, improving upon the previous best algorithm whose sampling complexity has at least cubic dependence on the dimensionality. Based upon our row sampling algorithm, we give the fastest known algorithm for quantile regression and a graph sparsification algorithm for balanced directed graphs. Our main technical contribution is to show that Lewis weights sampling, which has been used in row sampling algorithms for $\ell_p$ norms, can also be applied in row sampling algorithms for a variety of loss functions. We complement our theoretical results by experiments to demonstrate the practicality of our approach. 
\end{abstract}

\newcommand{\loss}{\phi}
\section{Introduction}
Linear regression is a fundamental problem in machine learning and statistics.
For a data matrix $A \in \RR^{n \times d}$ and a response vector $b \in \RR^n$ with $n \gg d$, the overconstrained linear regression problem can be formulated as solving the optimization problem
$\min_{x\in \RR^d} \loss(Ax - b)$
where $\loss : \RR^n \to \RR$ is a loss function.
Via the technique of row sampling, we have obtained remarkable speedups for solving linear regression for a wide range of loss functions.
Such techniques involve designing an importance sampling scheme and showing that if one samples the data matrix $A$ and the response vector $b$ accordingly, $\phi(Ax - b)$ is approximately preserved for all $x \in \RR^d$.
Thus, one can obtain an approximate solution to the original linear regression problem by solving a linear regression instance on the sampled data matrix and response vector, which is usually much smaller in size.

The technique of row sampling has received much attention in randomized numerical linear algebra and machine learning.
It is known that if one samples $O(d \log d / \varepsilon^2)$ rows of $A$ according the to leverage scores, the resulting matrix $A'$ would satisfy \[(1 - \varepsilon) \|Ax\|_2 \le \|A'x\|_2 \le (1 + \varepsilon)\|Ax\|_2\] for all $x \in \RR^d$ with high probability~\cite{tropp2012user}. 
Moreover, such row sampling algorithm can be implemented in $\widetilde{O}(\nnz(A) + \poly(d))$ time\footnote{Throughout the paper, we use $\widetilde{O}(f)$ to denote $O(f  \polylog f)$.}~\cite{clarkson2013low, meng2013low, nelson2013osnap, li2013iterative, cohen2015uniform}, where $\nnz(A)$ is the number of non-zero entries of $A$.
Building upon the result of Clarkson for $p = 1$~\cite{clarkson2005subgradient}, Dasgupta et al.\ obtain the first row sampling algorithm for general $\ell_p$ norms~\cite{dasgupta2009sampling}.
They show that by sampling $\widetilde{O}(d^{\max\{p/2+1,p\}+1} / \varepsilon^2)$ rows of $A$ according to the $\ell_p$ leverage scores,  the resulting matrix $A'$ satisfies \[(1 - \varepsilon) \|Ax\|_p \le \|A'x\|_p \le (1 + \varepsilon)\|Ax\|_p\] for all $x \in \RR^d$ simultaneously with high probability.
The algorithm in~\cite{dasgupta2009sampling} runs in $\widetilde{O}(nd^5)$ time, and the running time is further improved to $\widetilde{O}(\nnz(A) + \poly(d))$ in \cite{sohler2011subspace, clarkson2016fast, meng2013low, woodruff2013subspace, wang2019tight}.
Later, based on results in Banach space theory~\cite{talagrand1990embedding, talagrand1995embedding, bourgain1989approximation}, Cohen and Peng show that by sampling according to $\ell_p$ Lewis weights~\cite{cohen2015p}, the number of sampled rows can be further reduced to $\widetilde{O}(d / \varepsilon^2)$ when $p < 2$ and $\widetilde{O}(d^{p / 2} / \varepsilon^2)$ when $p \ge 2$,  which is nearly optimal as recently shown by Li et al.~\cite{li2020tight}.
Row sampling algorithms are also obtained for other loss functions, including $M$-estimators~\cite{clarkson2015sketching, clarkson2015input}, the Tukey loss function~\cite{CWW:tukey} and Orlicz norms~\cite{song2019efficient}.

Yang et al.\@ study row sampling algorithms for the quantile loss function~\cite{yang2013quantile}. Let $\tau\in (0,1]$ be a parameter and $y = Ax-b$. The quantile loss function is defined to be $h_{\tau}(y) = \sum_{i = 1}^n h_{\tau}(y_i)$, where
\begin{equation}\label{equ:quantile_def1}
h_{\tau}(y_i) = \begin{cases}
\tau y_i & \text{if $y_i \ge 0$} \\
(\tau - 1) y_i & \text{if $y_i \le 0$}
\end{cases}.
\end{equation}
The corresponding quantile regression problem is of particular interest, due to its wide applications in medicine \cite{cole1992smoothing, royston1994regression, royston2000goodness}, in hydrology \cite{pandey1999comparative}, in economics~\cite{koenker2001quantile, koenker1978regression}, in survival analysis~\cite{koenker2001reappraising}, etc.
See, e.g.,~\cite{koenker2001quantile, davino2013quantile} for surveys on theory and applications of quantile regression.
Quantile regression is also an active research topic in statistical machine learning.
See, e.g., \cite{chen2019quantile, chowdhury2019nonparametric, ota2019quantile, zheng2018high, zhao2017quantile} for recent developments. 
Using the notion of the $\ell_1$ leverage score \cite{dasgupta2009sampling} and known $\ell_1$ oblivious subspace embeddings~\cite{sohler2011subspace, meng2013low}, Yang et al.\ obtain the first quantile regression algorithm that runs in $\widetilde{O}(\nnz + \poly(d / \varepsilon))$ time~\cite{yang2013quantile}. 
More specifically, they show that by sampling $O(\tau/(1 - \tau)\cdot d^3/\varepsilon^2\big)$ rows according to the $\ell_1$ leverage scores, $h_\tau(Ax)$ is preserved for all $x \in \RR^d$ up to a factor of $1 \pm \varepsilon$.
Furthermore, quantile regression can be formulated as a linear program, and thus one can readily solve the smaller instance returned by the sampling process in $\poly(d / \varepsilon)$ time.

Although the results in \cite{yang2013quantile} are interesting from a theoretical point of view, the $d^3$ factor in the number of sampled rows required by their algorithm is far from being practical due to the rise of large-scale datasets.
Even for a moderate-sized dataset with $d = 300$, the algorithm in \cite{yang2013quantile} requires at least $d^3 = 2.7 \times 10^7$ samples, which is too large to gain any speedup over the na\"ive implementation without using row sampling. 
For modern large-scale datasets, e.g., the ImageNet dataset~\cite{imagenet_cvpr09}, $d$ is much larger than $10^5$ even after pre-processing.
In that case, even quadratic number of samples is too large to be practical. 
Thus, one may naturally ask whether it is possible to obtain a row sampling algorithm for the quantile loss function which requires only nearly linear number of samples, similar to the case of $\ell_p$ norms when $1 \le p \le 2$.
In this paper, we provide an affirmative answer to this question.

\paragraph{Our Contributions.}

We provide the first row sampling algorithm with nearly linear number of sampled rows for the quantile loss function. Note that we can write $Ax-b = A'x'$ for $A' = [A\ b]\in \RR^{n\times (d+1)}$ (concantenation of $A$ and $b$) and $x' = [\begin{smallmatrix} x \\ 1\end{smallmatrix}]\in \RR^{d+1}$, we overload the notations by renaming $A'$ to $A$ and $x'$ to $x$ and replacing $d$ with $d+1$, and consider instead preserving loss functions on $Ax$. For notational simplicity,  
we consider the equivalent problem of preserving 
\begin{equation}\label{equ:rho}
\rho_\tau(Ax) = \sum_{i = 1}^n \rho_\tau((Ax)_i),
\end{equation}
where $\tau\in[0,1]$ and $\rho_\tau : \RR \to \RR$ is defined as
\begin{equation}\label{equ:quantile_def2}
\rho_\tau(t) = \begin{cases}
				t &\quad t \geq 0\\
				-\tau t &\quad t < 0
			  \end{cases}.
\end{equation}
The problem of preserving $h_\tau(\cdot)$ is equivalent to that of preserving $\rho_\tau(\cdot)$, since 
 $h_{\tau}(t) = \tau \cdot \rho_{(1-\tau)/\tau}(t)$ for $\tau>0$.
Now we formally state our sampling result for the loss function $\rho_\tau(\cdot)$.

\begin{theorem}\label{thm:main_quantile}
There is an algorithm that receives a matrix $A \in \RR^{n \times d}$ and a parameter $\tau \in (0, 1]$, and outputs a matrix $\widetilde{A} \in \RR^{N \times d}$ with $N=\widetilde{O}(d/(\eps^2\tau^{2}))$ such that with probability at least $1 - 1 / \poly(d)$, for all $x \in \RR^d$, \[(1-\eps) \rho_{\tau}(Ax) \leq \rho_\tau(\widetilde{A}x)\leq (1+\eps)\rho_{\tau}(Ax).\]
The algorithm runs in $\widetilde{O}(\nnz(A) + d^{\omega})$ time.\footnote{We use $\omega$ to denote the matrix multiplication exponent. It is known that $\omega \le 2.373$~\cite{williams2012multiplying}.}
\end{theorem}
Our assumption that $\tau \le 1$ is without loss of generality, since $\rho_{1 / \tau}(t) = \frac{1}{\tau} \rho_{\tau}(-t)$, and thus our results can be directly translated to the case where $\tau > 1$.

Compared with~\cite{yang2013quantile}, the number of rows required by our row sampling algorithm has a nearly linear dependence on $d$, whereas the algorithm in~\cite{yang2013quantile} requires at least $O(d^3/(\varepsilon^2 \tau))$ sampled rows, in our notation. 
Thus, our result is better by a factor of $d^2$ when $\tau$ is constant (which is usually the case for applications in practice).
Furthermore, the running time of the row sampling algorithm in~\cite{yang2013quantile} is $\widetilde{O}(\nnz(A) + \poly(d))$, where the exponent hidden in the $\poly(d)$ term is much larger than $3$ due to the use of the ellipsoid rounding routine in~\cite{clarkson2016fast}. 


We then apply our sampling results to two problems: {\em quantile regression} and {\em balanced directed graph sparsification}.
\paragraph{Quantile Regression.}
For the quantile regression problem, notice that Theorem~\ref{thm:main_quantile} already implies an algorithm that runs in input-sparsity time.
Suppose that there is an algorithm that runs in time $\phi(s, d)$ for solving a quantile regression problem of size $s \times d$, then one can first use Theorem~\ref{thm:main_quantile} to reduce the problem size from $n \times d$ to $\widetilde{O}(d/\varepsilon^2\tau^2) \times d$, and thus the overall time complexity of the algorithm would be $\widetilde{O}(\nnz(A) + d^{\omega}) + \phi\left( \widetilde{O}(d/\varepsilon^2\tau^2), d\right)$.
As a comparison, the algorithm in~\cite{yang2013quantile} runs in $\widetilde{O}(\nnz(A) + \poly(d)) + \phi\big( \widetilde{O}(d^3 / \varepsilon^2\tau), d\big)$ time (see Theorem 1 in~\cite{yang2013quantile}).
Again, the exponent hidden in the $\poly(d)$ term is much larger than $3$.

To further demonstrate the flexibility and versatility of Theorem~\ref{thm:main_quantile}, we develop an algorithm for solving quantile regression using accelerated first-order methods. 
Here, our goal is to optimize the lower-order term that polynomially depends on $d$ in the time complexity while keeping the high-order term to be $\widetilde{O}(\nnz(A))$. To achieve this goal, we adopt the celebrated sketch-and-solve paradigm~\cite{woodruff2014sketching}, together with recently developed powerful tools for $\ell_1$ regression~\cite{durfee2018ell}.
As a result, we give the fastest known input-sparsity time algorithm for the quantile regression problem. 
\begin{theorem}[Informal version of Theorem~\ref{thm:reg}]
For a given $\tau\in (0,1]$, the quantile regression problem can be solved up to a factor of $(1 + \eps)$ in time $\widetilde{O}(\nnz(A) + d^{2.5} / (\eps^2 \tau^2))$.
\end{theorem}

\paragraph{Balanced Directed Graph Sparsification.}Another unique advantage of our sampling results over previous works is that, for constant $\eps$ and $\tau$, the number of sampled rows required by our algorithm is roughly linear in the dimensionality $d$. When $\tau = 0$, Theorem~\ref{thm:main_quantile} does not apply; nevertheless, a similar sampling scheme, when applied to the edge-vertex incidence matrix of a {\em balanced directed graph} (see Definition~\ref{def:balanced_graphs}), yields a graph sparsifier with nearly linear size. Recall Eqn.~\eqref{equ:rho} for loss functions applied to a vector. 

\begin{theorem}[Informal version of Corollary~\ref{cor:spar}]
For a given $\alpha$-balanced directed graph $G = (V, E, w)$ with $n$ vertices and $m$ edges, there is an algorithm that outputs $G' = (V, E', w')$ with $E' \subseteq E$ and $|E'| = O(n(\alpha/\eps)^2\log (n\alpha/\eps)))$ such that with probability at least $1-1/\poly(n)$, for all $x \in \RR^n$, 
\begin{equation}\label{equ:spar}
(1-\eps) \rho_0(Bx) \leq \rho_0(B'x)  \leq (1+\eps) \rho_0(Bx).
\end{equation}
Here $B$ is the edge-vertex incidence matrix of $G$ and $B'$ is the edge-vertex incidence matrix of $G'$.
\end{theorem}

We note that a sampling scheme gives meaningful sparsifiers for graphs of $n$ vertices only if the number of sampled rows is $o(n^2)$, since the largest possible number of edges in a graph is $O(n^2)$.
Hence the sampling scheme by Yang et al.~\cite{yang2013quantile}, for example, would sample $\Omega(n^3)$ rows and is not an option for graph sparsification.
Our sparsification result yields sparsifiers of roughly the same size as the state-of-the-art sparsification algorithm for balanced graphs by Ikeda and Tanigawa~\cite{IT:balanced_digraph}.\footnote{The two results are not directly comparable: for example, our bound has better dependence on the number of vertices $n$, while their bound has better dependence on the balance parameter $\alpha$.}
However, the sparsifier produced by~\cite{IT:balanced_digraph} is only guaranteed to be a cut sparsifier, meaning that \eqref{equ:spar} holds only for $x \in \{0, 1\}^n$ (see Remark~\ref{remark:spar}), whereas for our sparsifier, \eqref{equ:spar} holds for all $x \in \RR^n$, which is much stronger.
Furthermore, the algorithm in~\cite{IT:balanced_digraph} is specially tailored for balanced graphs, based on a number of highly combinatorial subroutines; in contrast, our algorithm follows a general-purpose sampling framework and utilizes the properties of balanced graphs in a rather black-box manner.
We believe our algorithm provides an alternative view of sparsification of (balanced) directed graphs, which also illustrates the power of our sampling scheme.

In Section~\ref{sec:exp}, we conduct an empirical evaluation of our row sampling algorithm for the quantile loss function to demonstrate the practicality of our approach. 

\paragraph{Organization.}
This paper is organized as follows. 
In Section~\ref{sec:pre}, we introduce necessary notations and definitions and also review known results regarding Lewis weights.
In Section~\ref{sec:quantile}, we present our main row sampling results.
In Section~\ref{sec:reg}, we present our algorithm for solving quantile regression.
In Section~\ref{sec:balanced_graph}, we give our results on balanced graph sparsification.
In Section~\ref{sec:extra}, we generalize our row sampling results in Section~\ref{sec:quantile} to a broader class of loss functions.
In Section~\ref{sec:exp}, we present our experimental results.
We conclude and discuss future directions in Section~\ref{sec:conclusion}.
Most technical proofs are deferred to the appendix.

\section{Preliminaries}\label{sec:pre}


\paragraph{Notations.}
We use $[n]$ to denote the set $\{1, 2, \ldots, n\}$.
For a vector $x \in \mathbb{R}^n$, we use $\|x\|_p$ to denote its $\ell_p$ norm i.e., $\|x\|_p = \left( \sum_{i = 1}^n |x_i|^p\right)^{1/ p }$.
For a matrix $A \in \mathbb{R}^{n \times d}$, we use $A^{\dagger}$ to denote its Moore-Penrose inverse.
For a matrix $A \in \mathbb{R}^{n \times d}$, we use $A_i$ to denote its $i$-th row, viewed as a column vector. The {\em leverage score} of $A_i$ is defined to be $\tau_{i}(A) = A_i^\top (A^\top A)^{\dagger}A_i $.
For two vectors $u$ and $v$, we use $\langle u, v \rangle$ to denote their inner product.
For a function $\phi : \mathbb{R} \to \mathbb{R}$ and a vector $y \in \mathbb{R}^n$, we use $\phi(y)$ to denote $\sum_{i = 1}^n \phi(y_i)$.

\paragraph{Graph Theory.}
We give some background on graph theory, which provides the context for our graph sparsification result.
In this paper, we focus on directed graphs $G = (V, E, w)$ with edge set $E \subseteq V \times V$, and edge weights $w : E \to \mathbb{R}_+$.
Throughout this paper, we assume that the directed graph $G$ is strongly connected.
For two sets of vertices $S$ and $T$, we use $w_G(S, T)$ to denote the total weight of edges from $S$ to $T$, i.e., 
$w_G(S, T) = \sum_{(u, v) \in (S \times T) \cap E} w(u, v)$.
We also say $w_G(S, T)$ is the {\em capacity} of the {\em directed cut} from $S$ to $T$.
We use $w(S, T)$ to denote $w_G(S, T)$ when $G$ is clear from the context.
For a directed graph $G = (V, E, w)$ with $n = |V|$ vertices and $m = |E|$ edges, the {\em edge-vertex incidence matrix} $B \in \mathbb{R}^{m \times n}$ is defined to be 
\[
B_{e, u} = \begin{cases}
w(e) & \text{if $e = (u, v)$} \\
-w(e) & \text{if $e = (v, u)$} \\
0 & \text{otherwise}
\end{cases}
.\]
In words, each row of $B$ corresponds to an edge $e = (u, v)$ in $G$, where the $u$-th entry is $w(e)$, $v$-th entry is $-w(e)$, and all other entries are $0$.

\paragraph{Lewis Weights.}

We now briefly review the definition and some basic properties of {\em Lewis weights} \cite{cohen2015p}, which play a key role in our sampling results.

\begin{definition}[$\ell_p$ Lewis weights {\cite{cohen2015p}}]\label{def:lewis_weights}
For a given matrix $A \in \mathbb{R}^{n \times d}$ and $p \ge 1$, 
we say $\{w_i\}_{i=1}^n$ are the {\em $\ell_p$ Lewis weights} of $A$ if they satisfy
$\tau_i(W^{1/2 - 1/p}A) = w_i$ for all $i \in [n]$,
where $W = \diag(w_1,\dots,w_n)$ is the $n\times n$ diagonal matrix whose $i$-th diagonal entry is $w_i$.
\end{definition}
We need the following fact regarding lewis weights proved in~\cite{cohen2015p}.
\begin{lemma}[\cite{cohen2015p}]\label{lem:sum_lewis}
The $\ell_p$ Lewis weights $\{w_i\}_{i=1}^n$ always exist and are unique, and $\sum_{i = 1}^n w_i \le d$.
\end{lemma}
The following theorem enables efficient (approximate) computation of Lewis weights, which we use as a subroutine in our sampling results.
\begin{theorem}[{\cite{cohen2015p}}]\label{thm:calc}
There is an algorithm that receives a matrix $A \in \mathbb{R}^{n \times d}$ and $p \ge 1$, and outputs $\{\overline{w}_i\}_{i=1}^n$, such that with probability at least $1 - 1 / \poly(d)$, for all $i \in [n]$, $w_i \le \overline{w}_i  \le 2 w_i$,
where $\{w_i\}_{i=1}^n$ are the $\ell_p$ Lewis weights of $A$.
The algorithm runs in $\widetilde{O}(\nnz(A) + d^{p / 2 + O(1)})$ time.
When $p < 4$, the algorithm runs in $\widetilde{O}(\nnz(A) + d^{\omega})$ time.
\end{theorem}

For graph sparsification specifically, we are interested in approximating Lewis weights for edge-vertex incidence matrices.
The following theorem states that in such cases, the running time stated in Theorem~\ref{thm:calc} can be further improved.
\begin{theorem}\label{thm:calc_graph}
When the matrix $A$ is the edge-vertex incidence matrix of a graph $G$ and $p < 4$, the algorithm in Theorem~\ref{thm:calc} runs in $\widetilde{O}(m)$ time, where $m$ is the number of edges in $G$.
\end{theorem}

\section{Row Sampling Based on Lewis Weights}\label{sec:quantile}
In this section, we present our row sampling results.
Our results are built upon the Lewis weights sampling framework introduced in~\cite{cohen2015p}.
In particular, given a matrix $A \in \mathbb{R}^{n \times d}$, they show that if one samples $O(d \log (d /\eps)/\eps^{2})$ rows of $A$ based on the $\ell_1$ Lewis weights (see Definition~\ref{def:lewis_weights}) and scales the rows properly, then one obtains a matrix $\widetilde{A}$ such that with high probability, for all $x \in \mathbb{R}^d$,
$(1 - \varepsilon) \|Ax\|_1 \le \|\widetilde{A}x\|_1 \le (1 + \varepsilon) \|Ax\|_1$.

In the following theorem, we show that for any function $\phi(\cdot)$ of the form $\phi(x) = a|x| + bx$ for some $0 \le b \le a$, sampling according to $\ell_1$ Lewis weights still approximately preserves the value of $\phi(Ax)$ for all $x \in \mathbb{R}^d$.
Our result generalizes that of~\cite{cohen2015p}, in the sense that if we set $a = 1$ and $b = 0$, we recover exactly their bound for the $\ell_1$ norm.
Later, based on Theorem~\ref{thm:quantile-sample}, we shall give a row sampling algorithm for the quantile loss function, and also provide an algorithm for producing sparsifiers for balanced directed graphs.

\begin{theorem}\label{thm:quantile-sample}
For a given matrix $A \in \mathbb{R}^{n \times d}$, let $\{w_i\}_{i=1}^n$ be its $\ell_1$ Lewis weights, and 
suppose the function $\phi(x) = a|x| + bx$ $(0 \le b \le a)$ satisfies that $\|Ax\|_1 \leq B\cdot \phi(Ax)$ for all $x \in \mathbb{R}^d$ for some $B>0$. 
There exists an absolute constant $C$ such that the following holds.

For any values $p_1,\ldots,p_n\geq 0$ such that $\sum_{i\in [n]} p_i = N$ and $p_i\ge C (aB/\eps)^{2} w_i\log N$,
if we generate a matrix $\widetilde{A}$ with $N$ rows, where each row is chosen independently as $A_i / p_i$ with probability $p_i / N$, then with probability  at least $1-1/\poly(d)$, it holds for all $x \in \mathbb{R}^d$ simultaneously that
\[(1-\eps) \phi(Ax)\leq \phi(\widetilde{A}x) \leq (1+\eps) \phi(Ax).\]
\end{theorem}
\begin{proof}[Proof (sketch)]
We follow the framework in~\cite{cohen2015p}. 
Suppose that $i_1,i_2,\dots,i_N$ are the indices chosen by the sampling process.
Our goal is to derive an upper bound on
\[
F =\!\!\!\! \sup_{x:  \phi(Ax) = 1} \left| \phi(\widetilde{A}x) - 1 \right| \!=\!\!\!\! \sup_{x: \phi(Ax) = 1} \left| \sum_{k=1}^N \frac{\phi(\langle A_{i_k}, x\rangle)}{p_{i_k}} \!-\! 1\right|.
\]
We shall show it holds for some $\ell$ that
\begin{equation}\label{eqn:target}
M = \E_{i_1,\dots,i_N} F^\ell  \leq \frac{\eps^\ell}{\poly(d)},
\end{equation}
and thus the result would follow from Markov's inequality. To prove the moment bound above, we first apply the symmetrization technique and obtain that
\[
M = \E_{i} \sup_{\phi(Ax) = 1} \left|\sum_{k=1}^N	\frac{\phi(\langle A_{i_k}, x\rangle)}{p_{i_k}} - 1\right|^\ell
\leq 2^\ell \E_{i,\sigma} \left[\left(\sup_{\phi(Ax)= 1}\left|\sum_{k=1}^N	\sigma_k\frac{\phi(\langle A_{i_k}, x\rangle)}{p_{i_k}}\right|\right)^\ell\right],
\]
where $i = (i_1, i_2, \ldots, i_N)$, and $\sigma=(\sigma_1,\sigma_2,\ldots, \sigma_N)$ is a Rademacher sequence.
Using the comparison theorem of Rademacher processes, we can show that
\[
\E_{\sigma} \left(\sup_{\phi(Ax) = 1}\left|\sum_{k=1}^N	\sigma_k\frac{\phi(\langle A_{i_k}, x\rangle)}{p_{i_k}}\right|\right)^\ell 
\leq 3a^\ell 2^{\ell-1} \E_{\sigma} \sup_{\phi(Ax) = 1} \left|\sum_{k=1}^N\sigma_k\frac{\langle A_{i_k}, x\rangle}{p_{i_k}}\right|^\ell.
\]
The supremum on the right-hand side is similar to the Rademacher process considered in~\cite{cohen2015p}. Analogously we can show that, conditioned on the event that for all $x \in \mathbb{R}^d$, $\|\widetilde{A}x\|_1\leq C_1\|Ax\|_1$, which holds with high probability for some absolute constant $C_1>0$, the expectation can be upper bounded as
\[
\E_\sigma \sup_{\phi(Ax) =  1} \left|\sum_{k=1}^N\sigma_k\frac{\langle A_{i_k}, x\rangle}{p_{i_k}}\right|^\ell 
\leq B^\ell \E_\sigma \sup_{\|Ax\|_1 =  1} \left|\sum_{k=1}^N\sigma_k\frac{\langle A_{i_k}, x\rangle}{p_{i_k}}\right|^\ell 
\leq B^\ell \frac{(\eps/(C'aB))^\ell}{\poly(d)}
\]
for some $\ell = \Theta(\log(N + d^2))$, where $C'$ is another absolute constant.
Thus \eqref{eqn:target} holds if the expectation is conditioned on the event defined above.
The failure probability of the event would be added to the overall failure probability, concluding the proof.
\end{proof}
Now Theorem~\ref{thm:main_quantile} immediately follows from Theorem~\ref{thm:quantile-sample}.
See Section~\ref{sec:proof} in the supplementary material for formal proofs.

We have successfully applied the Lewis weights sampling framework to provide row sampling algorithms for the quantile loss function.
In general, one may wonder if it is possible to apply Lewis weights sampling to more general loss functions. Cohen and Peng show that $\ell_p$ Lewis weights sampling works gracefully when the loss function is the $\ell_p$ norm~\cite{cohen2015p}.
In Section~\ref{sec:extra}, we show that when the loss function $\phi(x)$ can be approximated by $\|x\|_p^p$,  one can preserve $\phi(Ax)$ approximately for all $x \in \mathbb{R}^d$, by sampling proportional to $\ell_p$ Lewis weights and using a weighted version of the loss function $\phi(\cdot)$.


\newcommand{\wt}{\widetilde}
\newcommand{\Ab}{[A, b]}
\newcommand{\sAb}{[\wt{A}, \wt{b}]}

\section{Solving Quantile Regression}\label{sec:reg}
Based on the row sampling result in Theorem~\ref{thm:main_quantile} and recent advances for $\ell_1$ regression \cite{durfee2018ell}, we obtain a faster algorithm for solving the quantile regression problem.
The description of the algorithm is given in Algorithm~\ref{fig:alg}, and we prove its running time and correctness in Theorem~\ref{thm:reg}.
\begin{algorithm}
\begin{algorithmic}[1]
\STATE Obtain approximate $\ell_1$ Lewis weights for each row of $\Ab$ using Theorem~\ref{thm:calc}.
\STATE Sample $N = O\left( \frac{d}{\varepsilon^2 \tau^2}\log \frac{n}{\varepsilon \tau} \right)$ rows of $\Ab$ according to $\ell_1$ Lewis weights and obtain $\sAb \in \mathbb{R}^{N \times (d + 1)}$.
\STATE Let $\wt{A} = Q \cdot R$ be the QR-decomposition, where $Q\in\RR^{N\times d}$ and $R\in \RR^{d\times d}$. 
\STATE Run the Katyusha algorithm\footnotemark\ with $x_0 = (\wt{A}R^{-1})^T b$ to obtain $\overline{x}$ such that\label{alg:step:katyusha}
\begin{equation}\label{equ:approx_sketched}
\rho_\tau  \left(\wt{A}R^{-1}\overline{x}  -  \wt{b} \right) \le \left(1 + \frac{\varepsilon}{3}\right) \min_{x \in \mathbb{R}^d} \rho_\tau   \left(\wt{A}R^{-1}x - \wt{b} \right) .
\end{equation}
\STATE Return $x^* = R^{-1}\overline{x}$.
\end{algorithmic}
\caption{Algorithm for solving quantile regression}
\label{fig:alg}
\end{algorithm}
\footnotetext{Katyusha is an accelerated stochastic gradient descent algorithm by Allen-Zhu~\cite{allen2017katyusha} for minimization of a convex function. Here the convex function is $x\mapsto \rho_\tau(\widetilde{A}R^{-1}x-\widetilde{b})$ and we use the variant in Corollary 3.7 of \cite{allen2017katyusha}. The justification of applying this variant of Katyusha is given by Lemma~\ref{lem:conditions}(i).}

The following lemma, which is an easy corollary of Theorem~\ref{thm:main_quantile}, shows that if $\overline{x}$ satisfies the condition in \eqref{equ:approx_sketched}, then it is indeed an approximate solution to $\min_{x \in \mathbb{R}^d}\rho_\tau(Ax-b)$.
\begin{lemma}\label{lem:sketch}
For sufficiently small $\varepsilon > 0$, if $\overline{x}$ satisfies the condition in \eqref{equ:approx_sketched}, with probability at least $0.9$, the solution $x^* \in \mathbb{R}^d$ returned by Algorithm~\ref{fig:alg} satisfies \[\rho_\tau(Ax^*-b) \le (1 + \varepsilon) \min_{x \in \mathbb{R}^d}\rho_\tau(Ax-b).\]
\end{lemma}

Thus, we focus on obtaining an approximate solution to $\min_{x \in \mathbb{R}^d} \rho_\tau(\wt{A}R^{-1}x - \wt{b})$ hereinafter.
Our general strategy is to apply the accelerated SGD to $\min_{x \in \mathbb{R}^d} \rho_\tau(\wt{A}R^{-1}x - \wt{b})$. 
To this end, we formulate $\min_{x \in \mathbb{R}^d} \rho_\tau(\wt{A}R^{-1}x - \wt{b})$ as a finite-sum problem.
The following lemma shows that after applying Lewis weights sampling, with constant probability, each summand in the finite-sum problem has a bounded Lipschitz constant (see Definition~\ref{def:lip} in the supplementary material for the definition), and $x_0 = (\wt{A}R^{-1})^T b$ is close to the optimal solution.
\begin{lemma}\label{lem:conditions}
Let $x_0 = (\wt{A}R^{-1})^T b$ and $\wt{x^{\mathrm{opt}}} = \argmin_{x \in \mathbb{R}^d} \rho_\tau(\wt{A}R^{-1}x - \wt{b})$.
With probability at least $0.9$, 
\begin{enumerate}[label=(\roman*)]
\item the function $f(x) = \rho_\tau(\wt{A}R^{-1}x - \wt{b})$ can be written as $f(x) = \frac{1}{N} \sum_{i = 1}^N f_i(x)$, where each $f_i(\cdot)$ is $O(\sqrt{Nd})$-Lipschitz; and
\item it holds that
 \[ \|x_0 - \wt{x^{\mathrm{opt}}} \|_2 \le \sqrt{d/(N\tau^2)} \cdot \rho_\tau(\wt{A}R^{-1}\wt{x^{\mathrm{opt}}} - \wt{b}) \]
 and 
 \[ \rho_\tau(\wt{A}R^{-1}x_0 - \wt{b}) \le \sqrt{N/\tau^2}\cdot \rho_\tau(\wt{A}R^{-1}\wt{x^{\mathrm{opt}}} - \wt{b}). \]
 \end{enumerate}
\end{lemma}

Now we apply Katyusha, an accelerated SGD algorithm, from~\cite{allen2017katyusha}. The correctness and running time is summarized in the following lemma.

\begin{lemma}\label{lem:katyusha}
In Algorithm~\ref{fig:alg}, Step~4 
obtains $\overline{x}$ which satisfies that
\[
\rho_\tau(\wt{A}R^{-1}\overline{x} - \wt{b}) \le (1 + \varepsilon / 3)\min_{x \in \mathbb{R}^d} \rho_\tau(\wt{A}R^{-1}x - \wt{b})
\]
with probability at least $0.8$, and Step~4 
runs in $\wt{O}(d^{2.5}/(\tau^2 \varepsilon^2))$ time.
\end{lemma}

Now we give the efficiency and correctness of Algorithm~\ref{fig:alg}.
\begin{theorem}\label{thm:reg}
Algorithm~\ref{fig:alg} runs in time $\wt{O}(\nnz(A) + d^{2.5}/(\varepsilon^2 \tau^2))$ and with probability at least $0.7$ returns a vector $x^* \in \mathbb{R}^d$ such that 
\[\rho_\tau(Ax^*-b) \le (1 + \varepsilon) \cdot \min_{x \in \mathbb{R}^d}\rho_\tau(Ax-b).\]
\end{theorem}

We note that the failure probability of Algorithm~\ref{fig:alg} can be reduced to an arbitrarily small constant by independent repetitions and taking the best solution found among all repetitions.


\section{Balanced Graph Sparsification}\label{sec:balanced_graph}

In this section, we demonstrate how our sampling results in Section~\ref{sec:quantile} can be applied to graph sparsification.
In particular, we prove that our sampling scheme can be used to find a sparsifier of {\em balanced} directed graphs, a notion defined in~\cite{EMPS16} as follows. 
\begin{definition}\label{def:balanced_graphs}
For a strongly connected directed graph $G = (V, E, w)$ and $\alpha \ge 1$, we say $G$ is \emph{$\alpha$-balanced} if for every subset $S\subset V$, we have
$
w(S, V \setminus S) \le \alpha \cdot w(V \setminus S, S)
$.
\end{definition}
Recall that $w(S, T)$ is the capacity of the directed cut between $S$ and $T$.
Therefore, the preceding definition essentially says that a graph is balanced if the outgoing capacity from any set of vertices $S$ is roughly the same (up to a factor of $\alpha$) as the ingoing capacity to $S$.
The parameter $\alpha$ measures how balanced the graph is.
In particular, any undirected graph is $1$-balanced.

Our plan is to apply our sampling results on the edge-vertex incidence matrix of the balanced graph.
In order to do that, we need first to relate balanced graphs to the quantile loss function, and show that the incidence matrix indeed satisfies the conditions required by our sampling results.
In particular, we shall need the following lemma regarding balanced directed graphs.
\begin{lemma}\label{lem:balanced_graph}
For a given strongly connected directed graph $G = (V, E, w)$ with $n$ vertices and $m$ edges, let $B \in \mathbb{R}^{m \times n}$ be the edge-vertex incidence matrix of $G$.
If $G$ is $\alpha$-balanced, then it holds for all $x \in \mathbb{R}^n$ that
$\|Bx\|_1 \leq (1+\alpha)\rho_0(Bx)$.
\end{lemma}

Now we are ready to apply our sampling schemes to balanced graphs.
As another application of Theorem~\ref{thm:quantile-sample}, we give a randomized construction of a sparsifier of $O(n(\alpha/\eps)^2\log (\alpha n/\eps)))$ edges for $\alpha$-balanced graphs.
The algorithm is given in Algorithm~\ref{fig:spar}, and we prove its correctness and running time in Corollary~\ref{cor:spar}.

\begin{algorithm}
\begin{algorithmic}[1]
\STATE Obtain approximate $\ell_1$ Lewis weights for each row of the edge-vertex incidence matrix $B$ of the given $\alpha$-balanced graph $G$ using Theorem~\ref{thm:calc_graph}.
\STATE Sample $m' = O(n(\alpha/\eps)^2\log (\alpha n/\eps)))$ rows of $B$ according to approximate $\ell_1$ Lewis weights and obtain $B' \in \mathbb{R}^{m' \times n}$.
\STATE Return the directed graph $G'$ that corresponds to the edge-vertex incidence matrix $B'$. 
\end{algorithmic}
\caption{Algorithm for sparsifying balanced directed graphs.}
\label{fig:spar}
\end{algorithm}


\begin{corollary}\label{cor:spar}
For a given $\alpha$-balanced strongly connected directed graph $G = (V, E, w)$ with $n$ vertices and $m$ edges, let $B \in \mathbb{R}^{m \times n}$ be the edge-vertex incidence matrix of $G$.
Algorithm~\ref{fig:spar} outputs $G' = (V, E', w')$ with $E' \subseteq E$ and $|E'|  = m' = O(n(\alpha/\eps)^2\log (\alpha n/\eps)))$ in $\wt{O}(m)$ time.
Let $B' \in \mathbb{R}^{m' \times n}$ be the edge-vertex incidence matrix of $G'$.
With probability at least $1-1/\poly(n)$, for all $x \in \mathbb{R}^d$,
\[(1-\eps) \rho_0(Bx) \leq \rho_0(B'x)  \leq (1+\eps) \rho_0(Bx).\]
\end{corollary}

\begin{remark}\label{remark:spar}
The sparsifier $G'$ obtained in Corollary~\ref{cor:spar} is a {\em cut sparsifier}, meaning that it approximately preserves the capacities of all cuts of $G$.
To see this, for any given $S \subseteq V$, let $x_S \in \{0,1\}^n$ be the indicator vector of $S$, i.e.,
$
(x_S)_i = \left\{\begin{smallmatrix}
1 & i \in S\\
0 & i \notin S
\end{smallmatrix}\right.
$,
then $\rho_0(Bx) = w_G(S, V \setminus S)$.
Similarly, $\rho_0(B'x) = w_{G'}(S, V \setminus S)$.
Since for all $x \in \mathbb{R}^d$, $(1-\eps) \rho_0(Bx) \leq \rho_0(B'x)  \leq (1+\eps) \rho_0(Bx)$, it follows that for all $S \subseteq V$, 
$(1-\eps) w_G(S, V \setminus S) \leq w_{G'}(S, V \setminus S) \leq (1+\eps) w_G(S, V \setminus S).$
\end{remark}
%

\section{Row Sampling for Loss Functions Approximated by $\ell_p$ Norms}\label{sec:extra}
In Section~\ref{sec:quantile}, we have shown that $\ell_1$ Lewis weights sampling approximately preserves the quantile loss function $\rho_\tau(Ax)$ for all $x \in \mathbb{R}^d$. A natural question is whether it is possible to preserve other loss functions using $\ell_p$ Lewis weights sampling.
The following theorem shows that it is indeed possible.
However, owing to the lack of homogeneity in the loss function, we have to use a weighted version of the loss function as the estimator. 
The proof is similar to that of Theorem~\ref{thm:quantile-sample} and is provided in Section~\ref{sec:extra_proof} in the supplementary material.

\begin{theorem}\label{thm:generalized-sample}
Let $p\geq 1$. Suppose $\phi : \mathbb{R} \to \mathbb{R}$ satisfies $|\phi(x)-\phi(y)|\leq L\left| |x|^p - |y|^p\right|$ for all $x, y\in \RR$ for some $L > 0$ and $\phi(x)\geq \gamma|x|^p$ for all $x\in\RR$ for some $\gamma > 0$. 

Let matrix $A\in\RR^{n\times d}$ and $\{w_i\}_{i=1}^n$ be its $\ell_p$ Lewis weights. 
For any values $p_1,\ldots,p_n\geq 0$ such that $\sum_{i\in [n]} p_i = N$ and $
p_i\ge f\left(p,d,\frac{\gamma}{L}\eps, \delta\right) w_i$, the following holds.

Let $i_1,\dots,i_N$ be i.i.d.\@ random variables such that $\Pr[i_k = j] = p_j/N$ for all $k \in [N]$. Define $w_i = 1/p_{i_k}$. Construct $\widetilde{A} \in \mathbb{R}^{N \times d}$ so that $\widetilde{A}_k = A_{i_k}$ for each $k \in [N]$.
With probability at least $1-\delta$, it holds for all $x \in \mathbb{R}^d$ simultaneously that
\[
(1-\eps)\phi(Ax) \leq \phi^{w}(\widetilde{A}x) \leq (1+\eps)\phi(Ax),
\]
where $\phi^w(y) = \sum_{i = 1}^N w_i \phi(y_i)$ for $y \in \mathbb{R}^N$.

Bounds on $f(p, d,\eps, \delta)$ are given in the table below, where $C_p$ is a constant depending only on $p$ and $C$ is an absolute constant.

\begin{center}
\begin{tabular}{ccc}
\hline
$p$ & $\delta$ & $f(p,d,\eps, \delta)$ \\
\hline
$p = 1$ 		& $1/\poly(d)$ & $C(1/\eps^2)\log (d/\eps)$ \\
$1 < p < 2$ 	& constant 				& $C(1/\eps^2)\log(d/\eps)\log^2 \log(d/\eps)$ \\
$p > 2$ 		& $1/\poly(d)$ & $ C_p(1/\eps^2)d^{p/2}\log^2 d\log(d/\eps)$  \\
\hline
\end{tabular}
\end{center}
\end{theorem}

Therefore, to obtain a row sampling algorithm for loss functions that satisfy the conditions in Theorem~\ref{thm:generalized-sample}, we first invoke the algorithm in Theorem~\ref{thm:calc} to obtain approximate $\ell_p$ Lewis weights of the input matrix $A$ and then sample rows of $A$ according to Theorem~\ref{thm:generalized-sample}. The total number of sampled rows $N = \sum_{i = 1}^n p_i = O(d \cdot f(p,d,\eps, \delta))$ by Lemma~\ref{lem:sum_lewis}.

For solving linear regression, $\phi(\cdot)$ is usually convex, in which case $\phi^{w}(\cdot)$ is also convex. 
Therefore, to obtain an approximate solution to the linear regression problem $\min_{x \in \mathbb{R}^d} \phi(Ax - b)$, one just needs to solve a much smaller instance $\min_{x \in \mathbb{R}^d} \phi^{w}(\widetilde{A} x - \widetilde{b})$, which can still be formulated as a convex program and thus can be efficiently solved.

We remark that by using the notion of $\ell_p$ leverage score and techniques in \cite{dasgupta2009sampling}, we may obtain a results similar to Theorem~\ref{thm:generalized-sample}. However, the number of samples required by such an approach will be much larger than that of Theorem~\ref{thm:generalized-sample}.

\begin{figure*}[!htb]
\captionsetup[subfigure]{width=0.35\textwidth,farskip=0pt,captionskip=-5pt}

\subfloat[$\tau=0.5$, $\|x\!-\!x^\ast\|_2/\|x^\ast\|_2$]{%
\includegraphics[clip, trim = 3.4cm 10.8cm 4cm 11.05cm, width=0.32\textwidth]{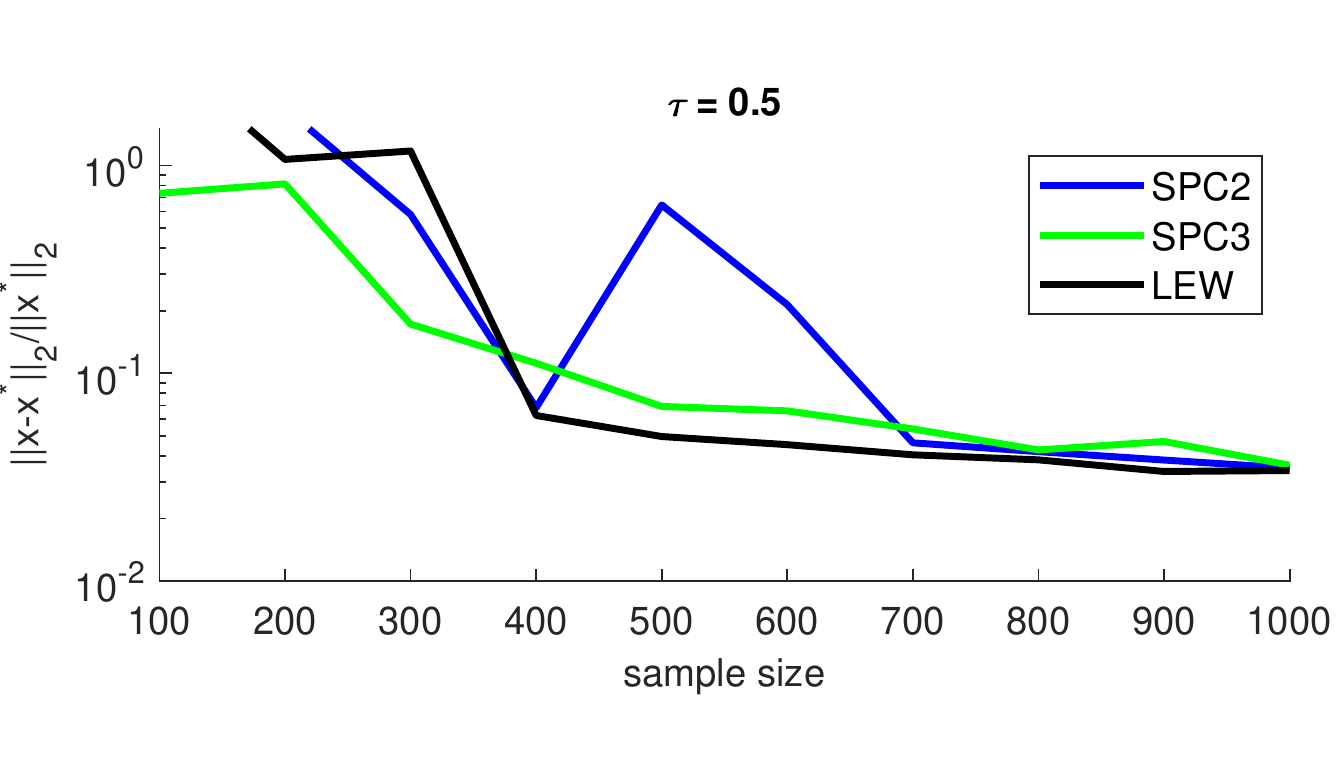}
}
\hfill
\subfloat[$\tau=0.75$, $\|x\!-\!x^\ast\|_2/\|x^\ast\|_2$]{%
\includegraphics[clip, trim = 3.4cm 10.8cm 4cm 11.05cm, width=0.32\textwidth]{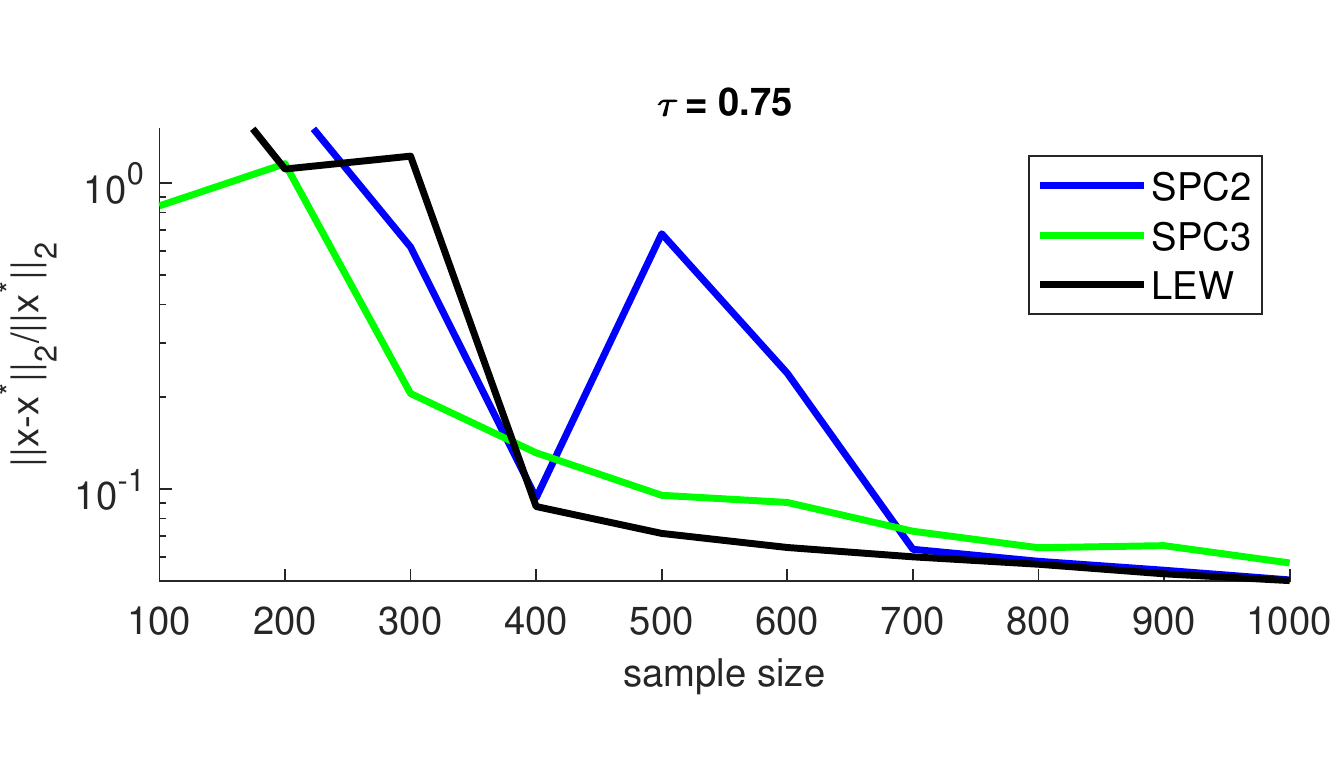}
}
\hfill
\subfloat[$\tau=0.95$, $\|x\!-\!x^\ast\|_2/\|x^\ast\|_2$]{%
\includegraphics[clip, trim = 3.4cm 10.8cm 4cm 11.05cm, width=0.32\textwidth]{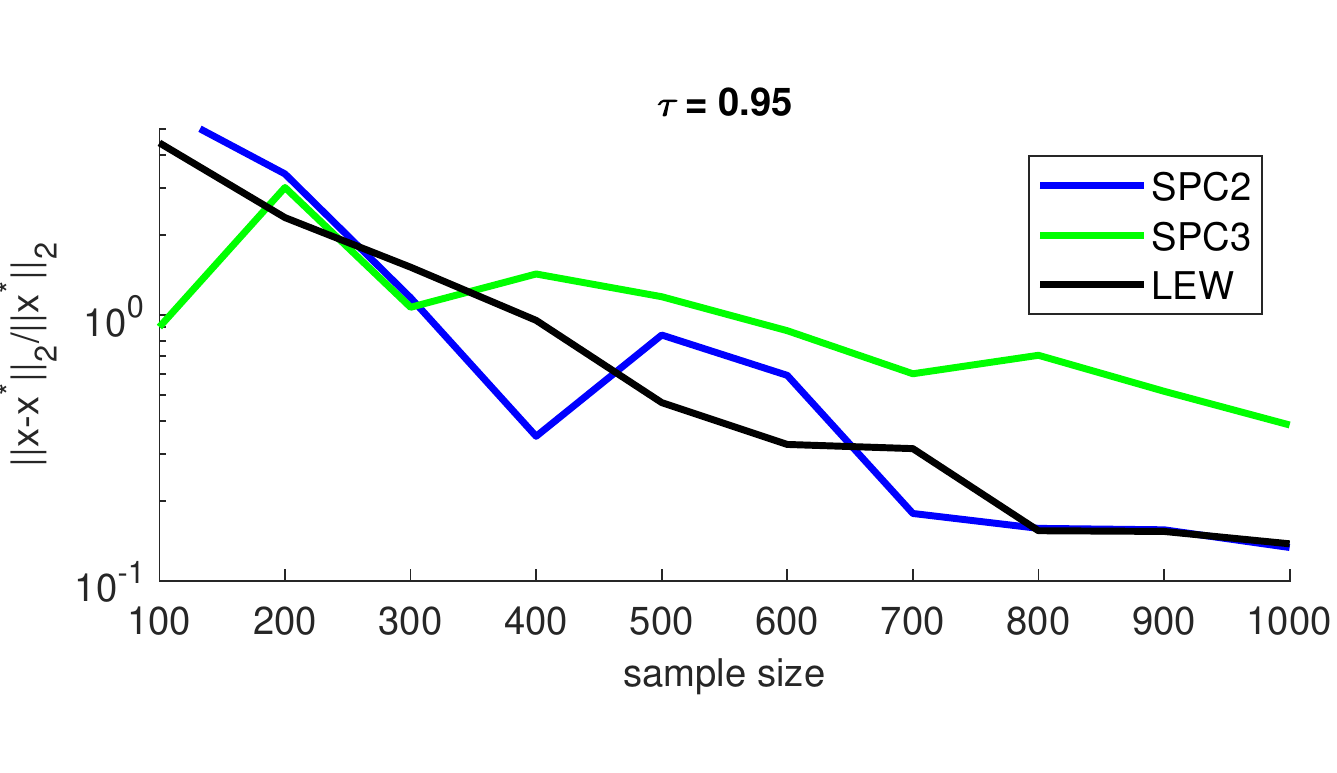}
}

\subfloat[$\tau=0.5$, $\|x\!-\!x^\ast\|_1/\|x^\ast\|_1$]{%
\includegraphics[clip, trim = 3.4cm 10.8cm 4cm 11.05cm, width=0.32\textwidth]{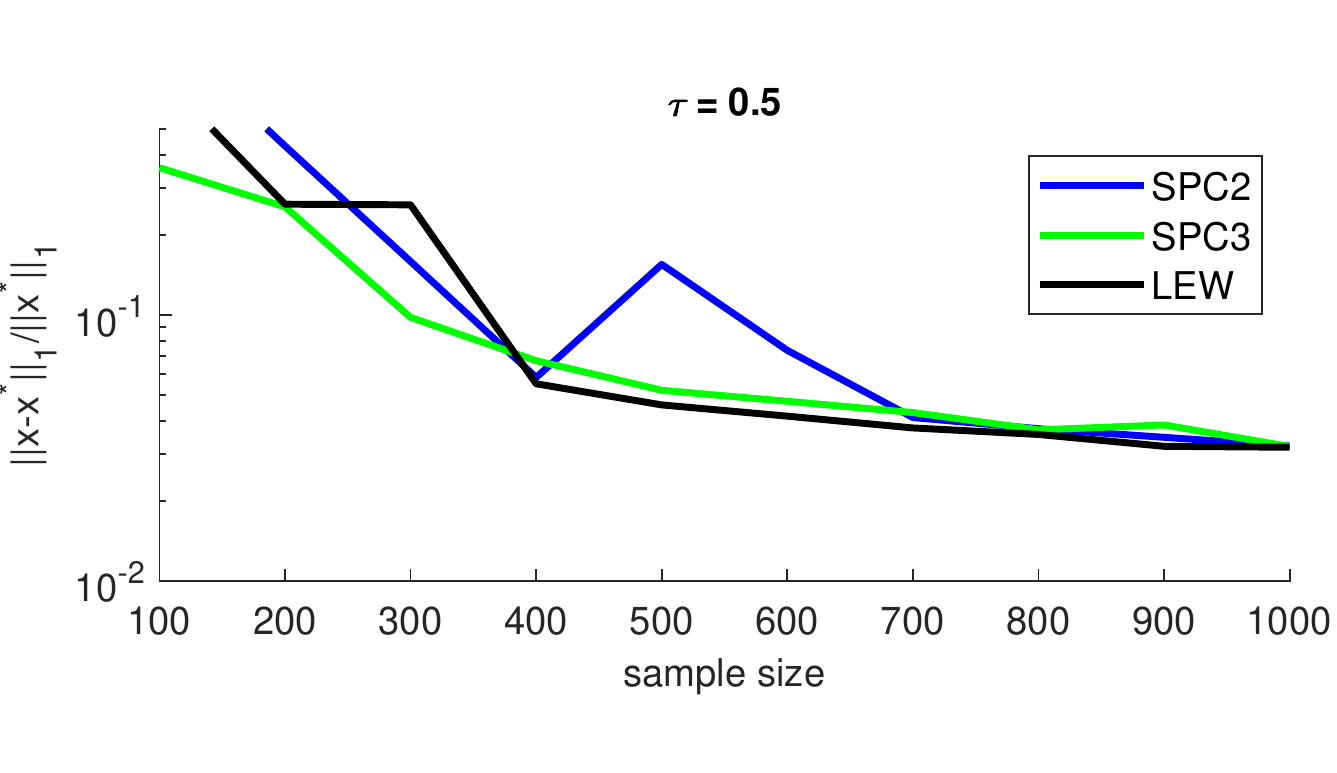}
}
\hfill
\subfloat[$\tau=0.75$, $\|x\!-\!x^\ast\|_1/\|x^\ast\|_1$]{%
\includegraphics[clip, trim = 3.4cm 10.8cm 4cm 11.05cm, width=0.32\textwidth]{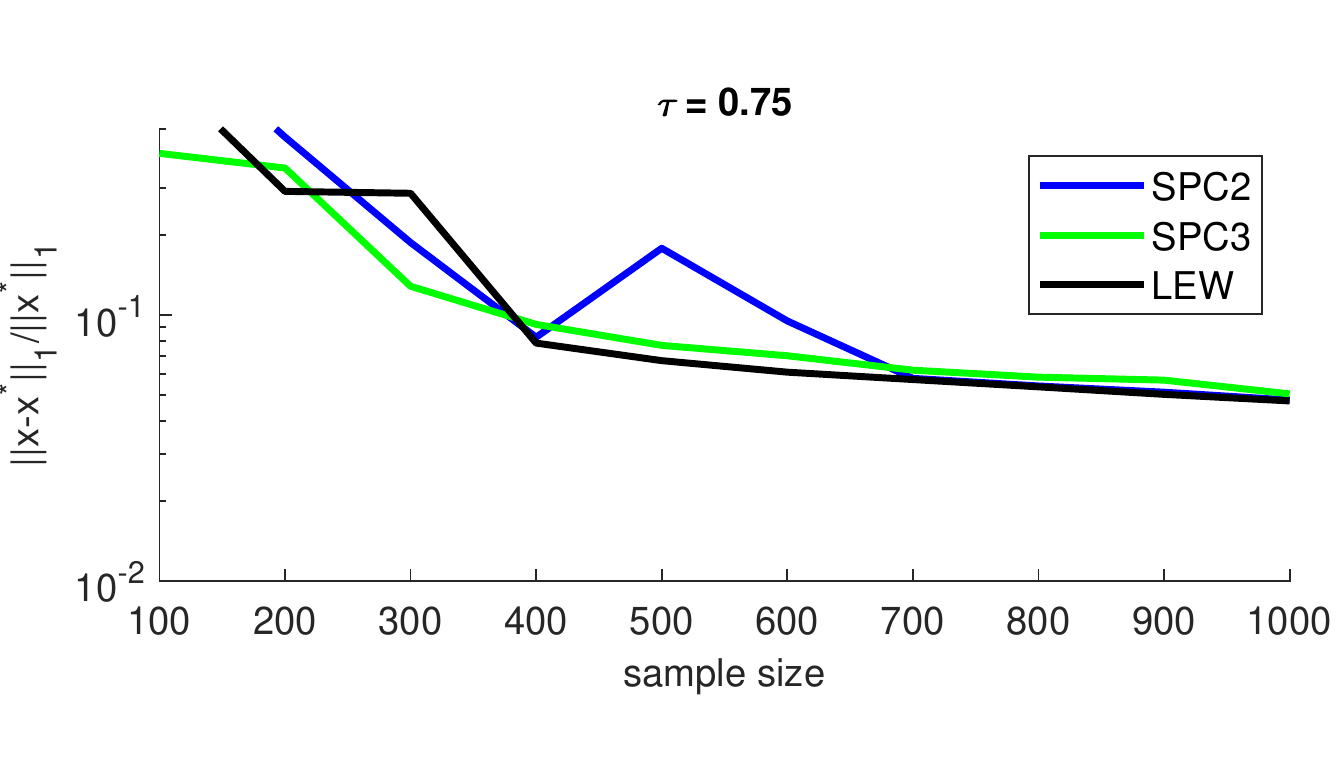}
}
\hfill
\subfloat[$\tau=0.95$, $\|x\!-\!x^\ast\|_1/\|x^\ast\|_1$]{%
\includegraphics[clip, trim = 3.4cm 10.8cm 4cm 11.05cm, width=0.32\textwidth]{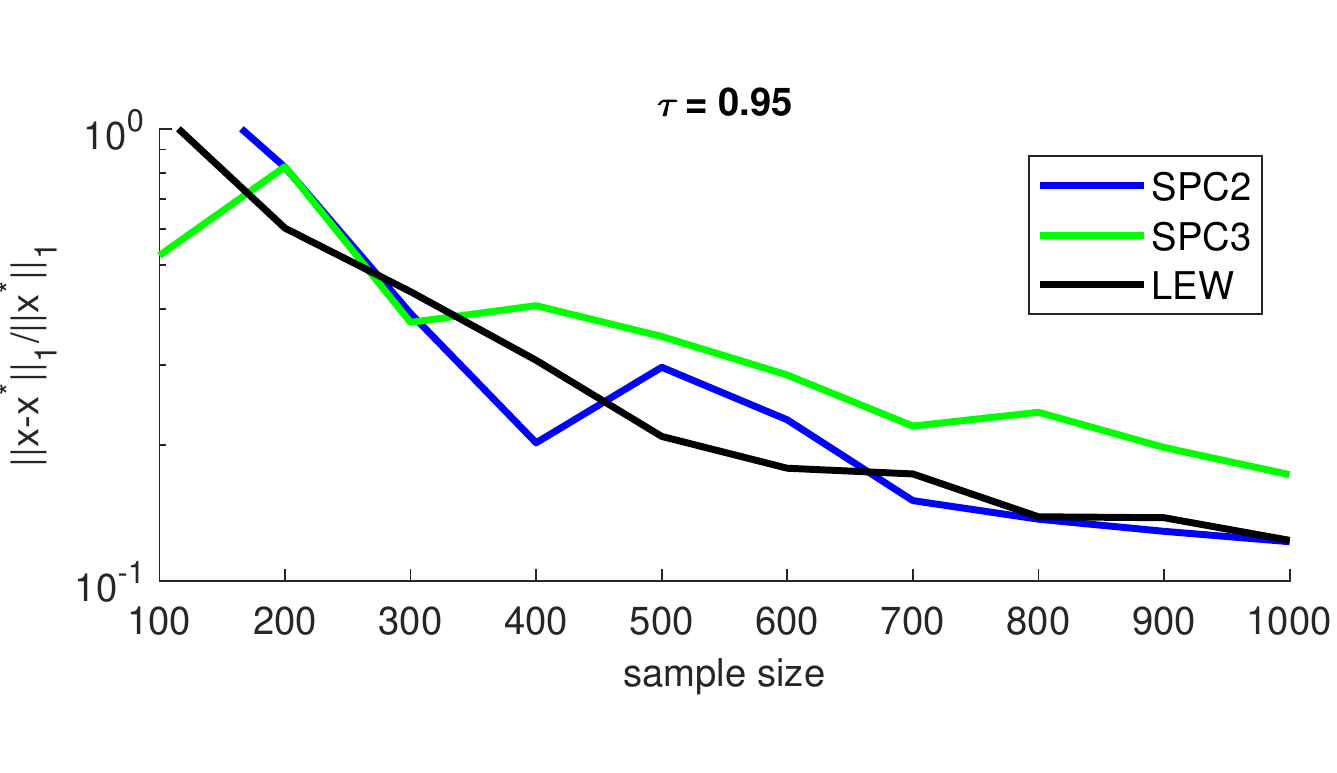}
}

\subfloat[$\tau=0.5$, $\|x\!-\!x^\ast\|_\infty/\|x^\ast\|_\infty$]{%
\includegraphics[clip, trim = 3.4cm 10.8cm 4cm 11.05cm, width=0.32\textwidth]{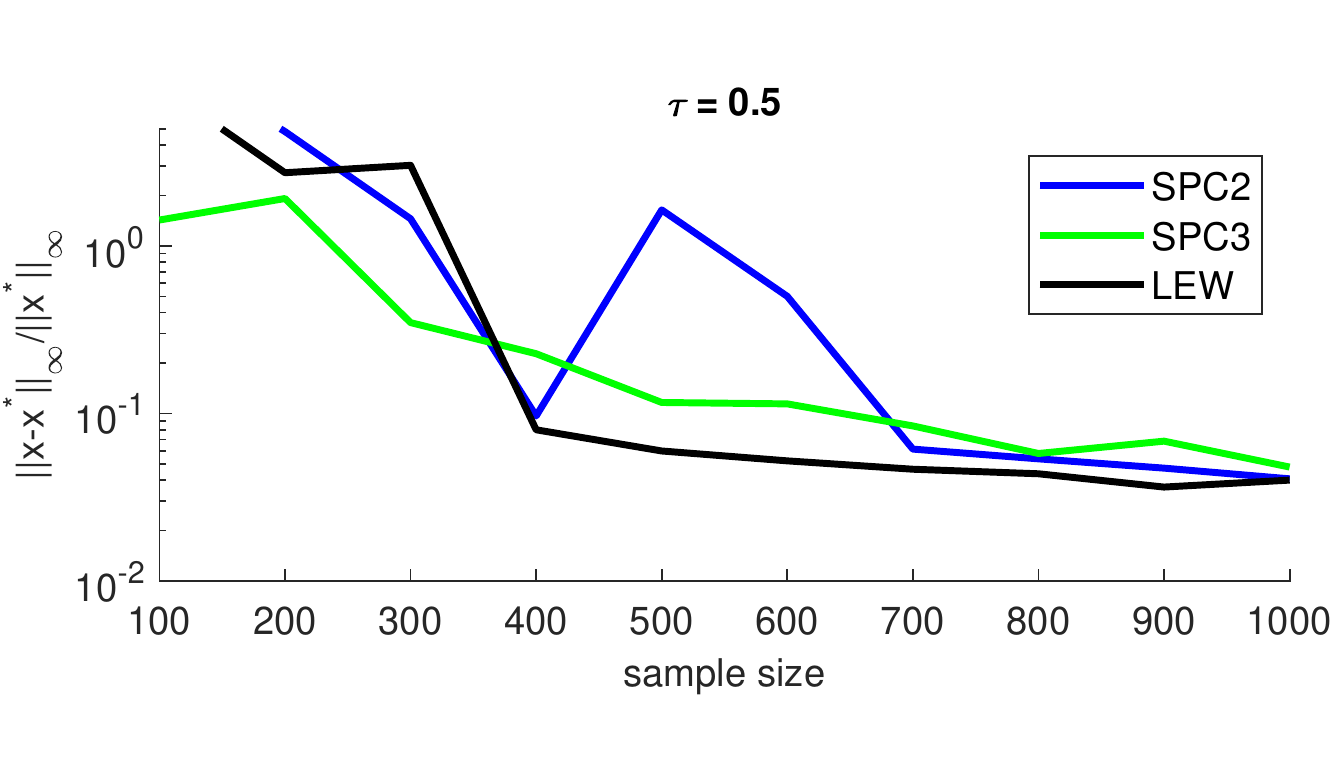}
}
\hfill
\subfloat[$\tau=0.75$, $\|x\!-\!x^\ast\|_\infty/\|x^\ast\|_\infty$]{%
\includegraphics[clip, trim = 3.4cm 10.8cm 4cm 11.05cm, width=0.32\textwidth]{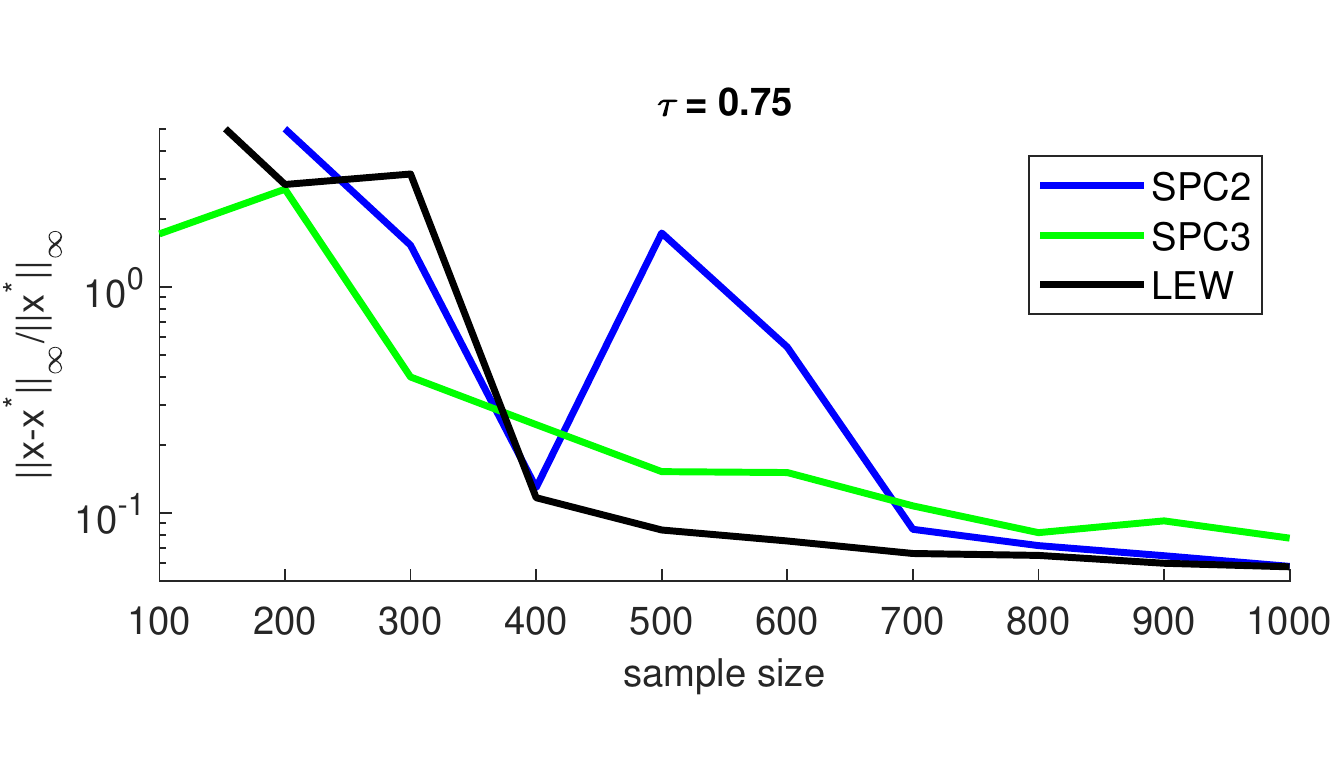}
}
\hfill
\subfloat[$\tau=0.95$, $\|x\!-\!x^\ast\|_\infty/\|x^\ast\|_\infty$]{%
\includegraphics[clip, trim = 3.4cm 10.8cm 4cm 11.05cm, width=0.32\textwidth]{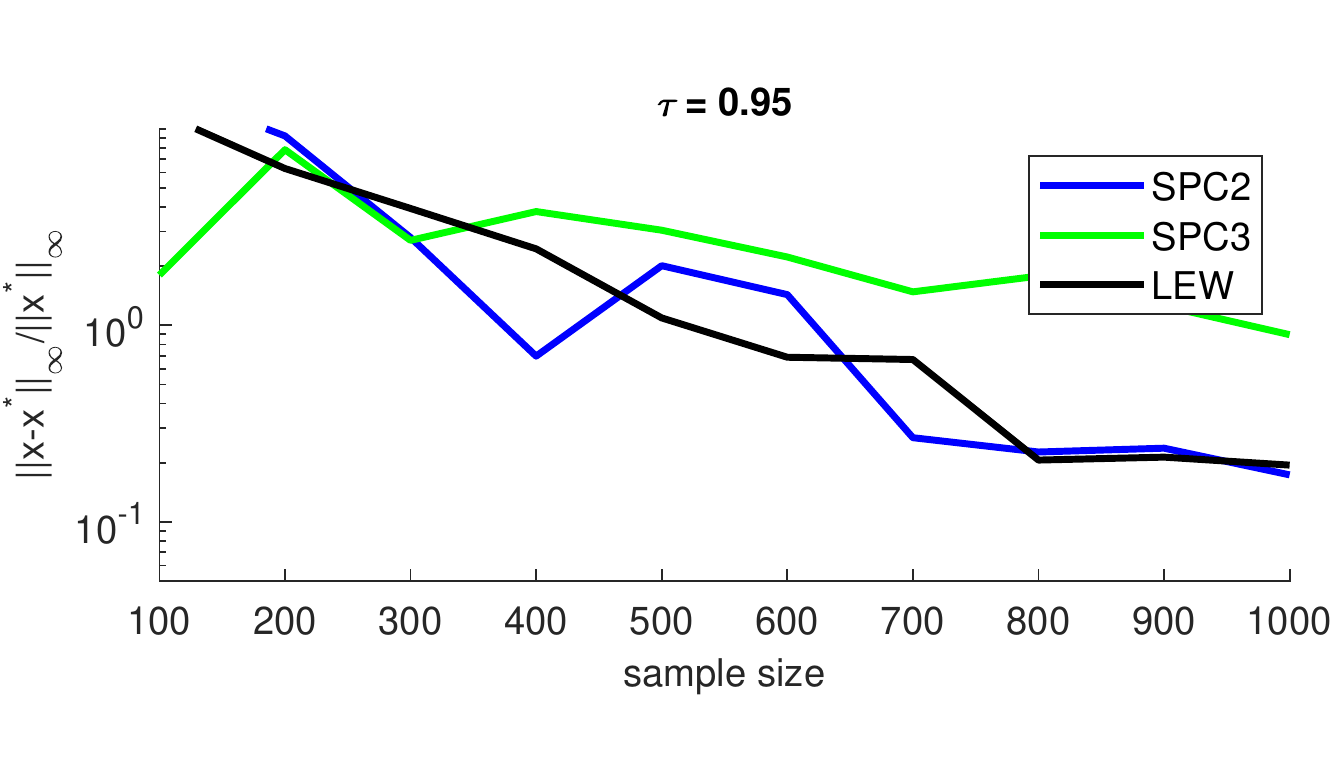}
}

\caption{\small Mean of the relative errors of the solution vector (measured in three different norms, the $\ell_2$ norm, the $\ell_1$ norm and the $\ell_\infty$ norm) of $3$ different methods on synthetic data. The plotted lines represent the mean of the errors of $50$ independent trials.}\label{fig:vary_s}
\end{figure*}

\begin{figure*}[!htb]
\captionsetup[subfigure]{width=0.3\textwidth,farskip=0pt,captionskip=-10pt}

\subfloat[$\tau=0.5$, $\|x\!-\!x^\ast\|_2/\|x^\ast\|_2$]{%
\includegraphics[clip, trim = 3.3cm 10.8cm 4.1cm 11.05cm, width=0.32\textwidth]{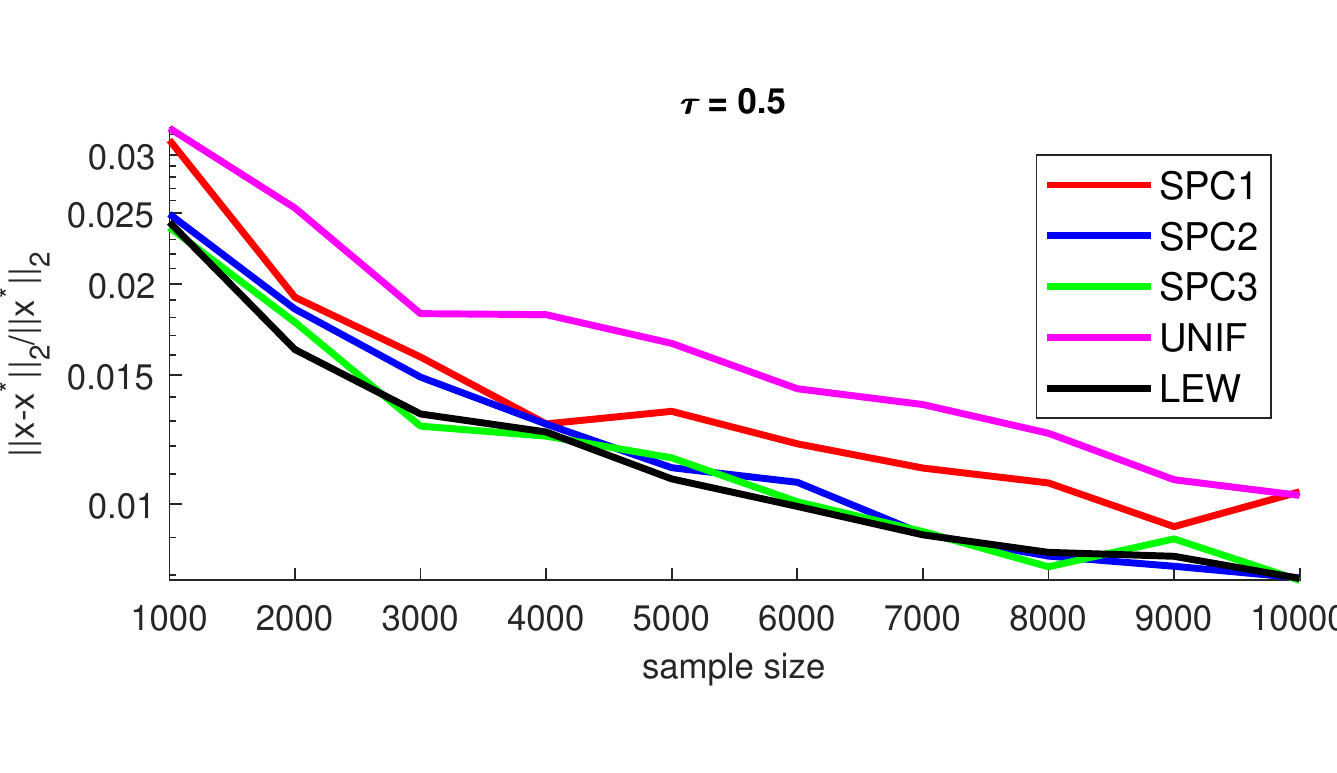}
}
\hfill
\subfloat[$\tau=0.75$, $\|x\!-\!x^\ast\|_2/\|x^\ast\|_2$]{%
\includegraphics[clip, trim = 3.45cm 10.8cm 4.1cm 11.05cm, width=0.32\textwidth]{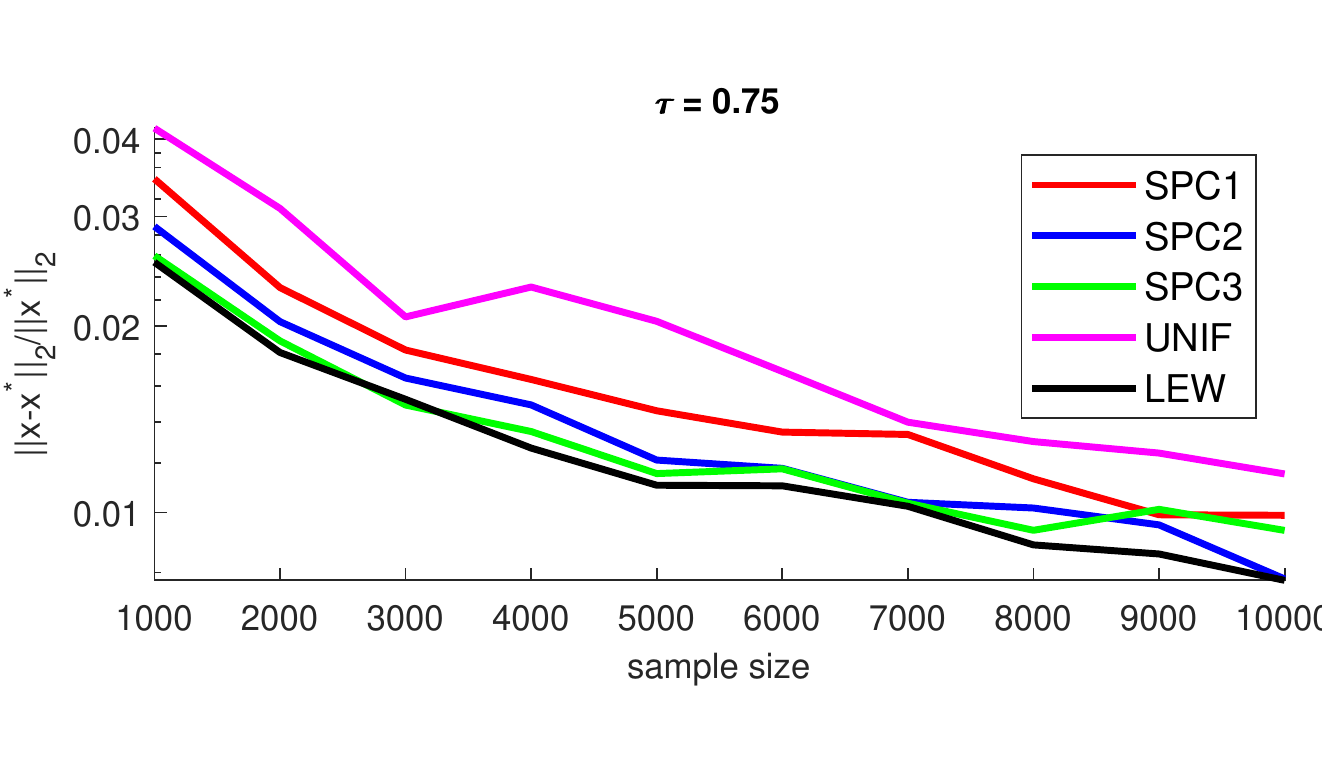}
}
\hfill
\subfloat[$\tau=0.95$, $\|x\!-\!x^\ast\|_2/\|x^\ast\|_2$]{%
\includegraphics[clip, trim = 3.45cm 10.8cm 4.1cm 11.05cm, width=0.32\textwidth]{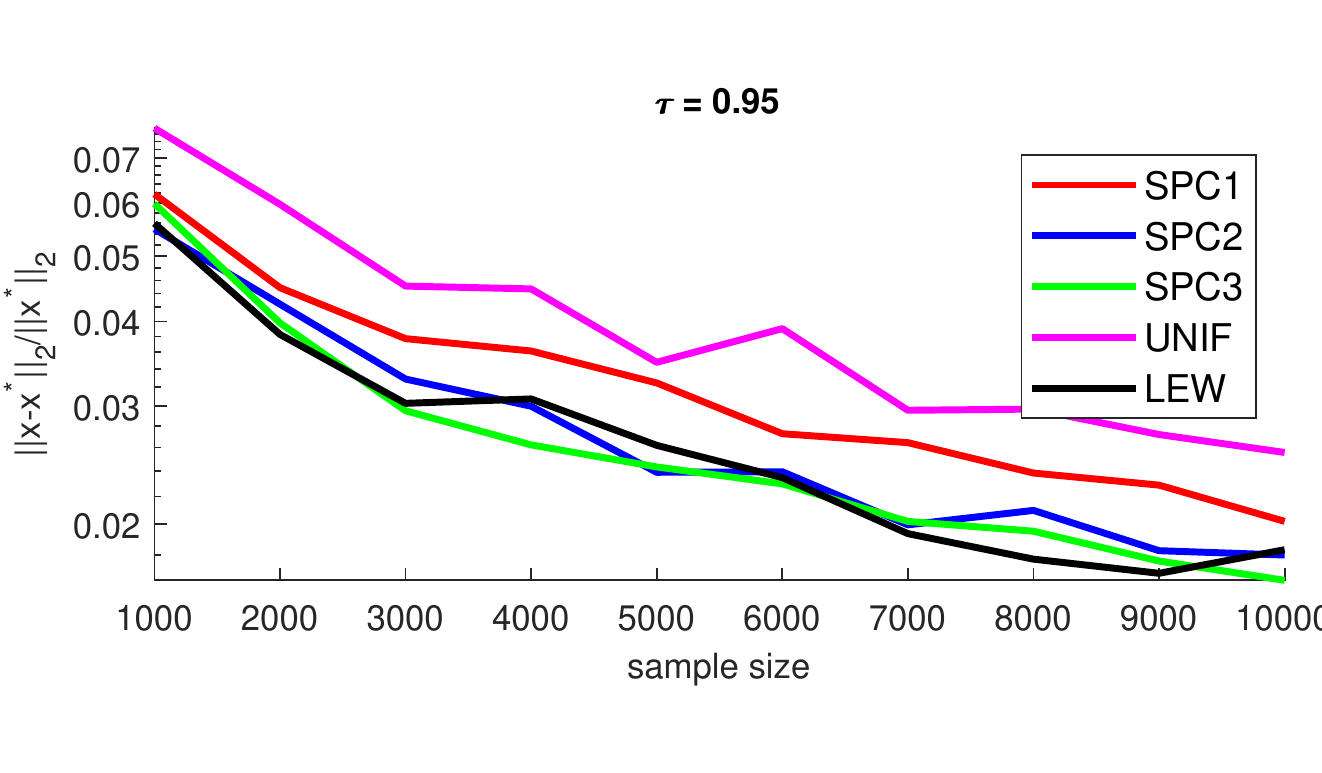}
}

\subfloat[$\tau=0.5$, $\|x\!-\!x^\ast\|_1/\|x^\ast\|_1$]{%
\includegraphics[clip, trim = 3.45cm 10.8cm 4.1cm 11.05cm, width=0.32\textwidth]{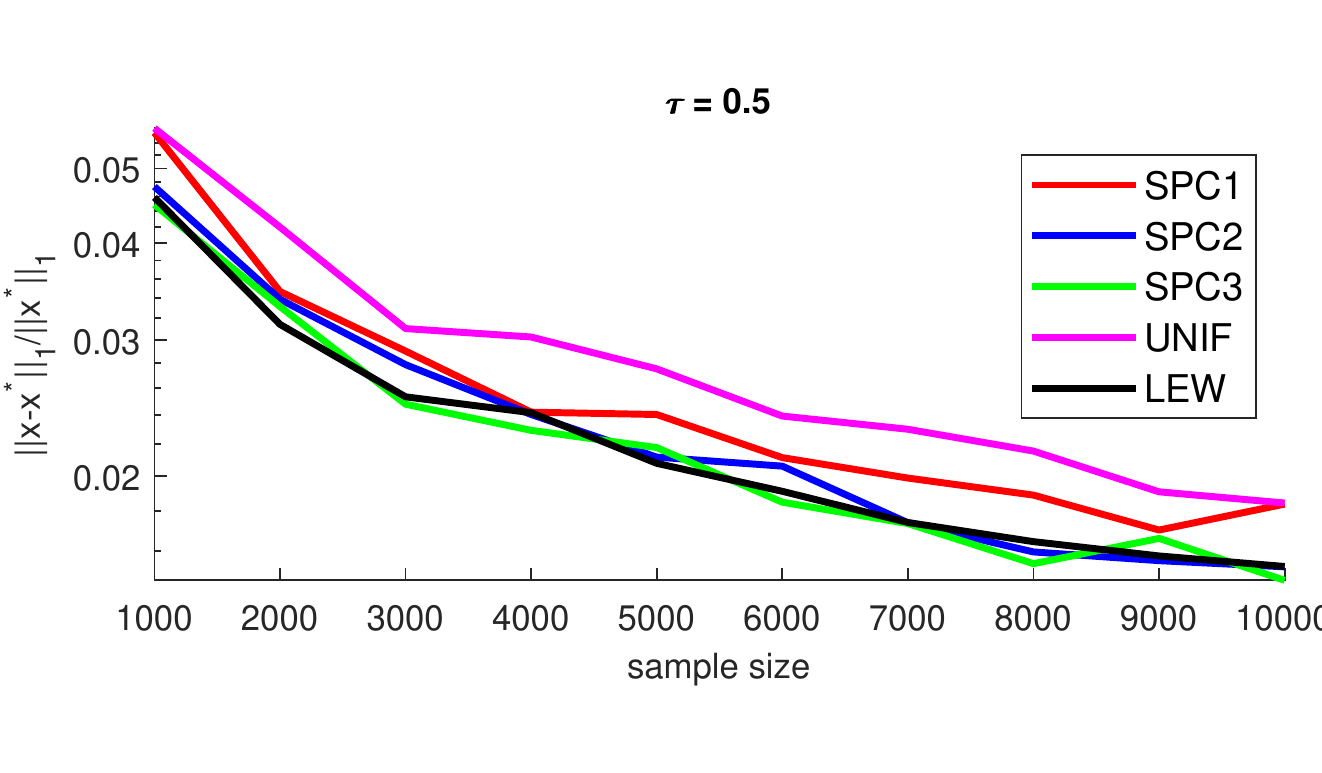}
}
\hfill
\subfloat[$\tau=0.75$, $\|x\!-\!x^\ast\|_1/\|x^\ast\|_1$]{%
\includegraphics[clip, trim = 3.45cm 10.8cm 4.1cm 11.05cm, width=0.32\textwidth]{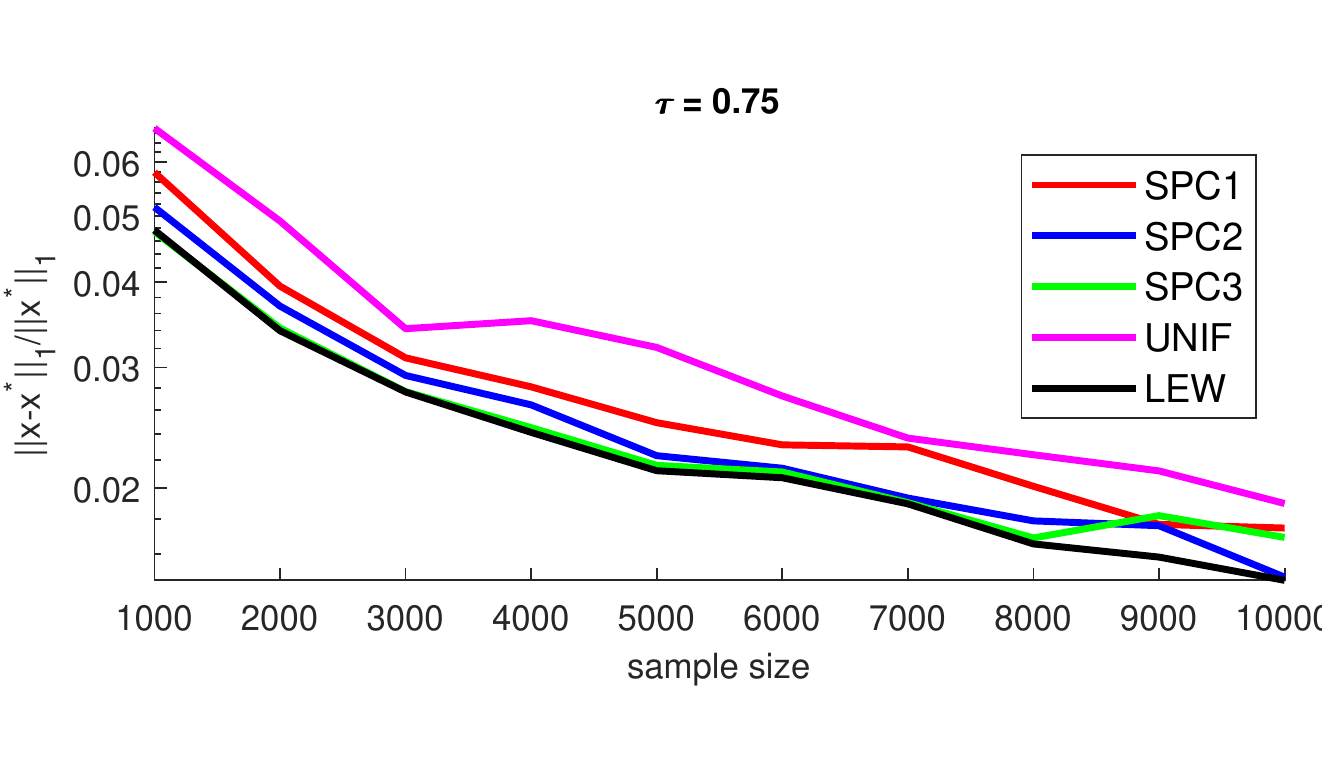}
}
\hfill
\subfloat[$\tau=0.95$, $\|x\!-\!x^\ast\|_1/\|x^\ast\|_1$]{%
\includegraphics[clip, trim = 3.45cm 10.8cm 4.1cm 11.05cm, width=0.32\textwidth]{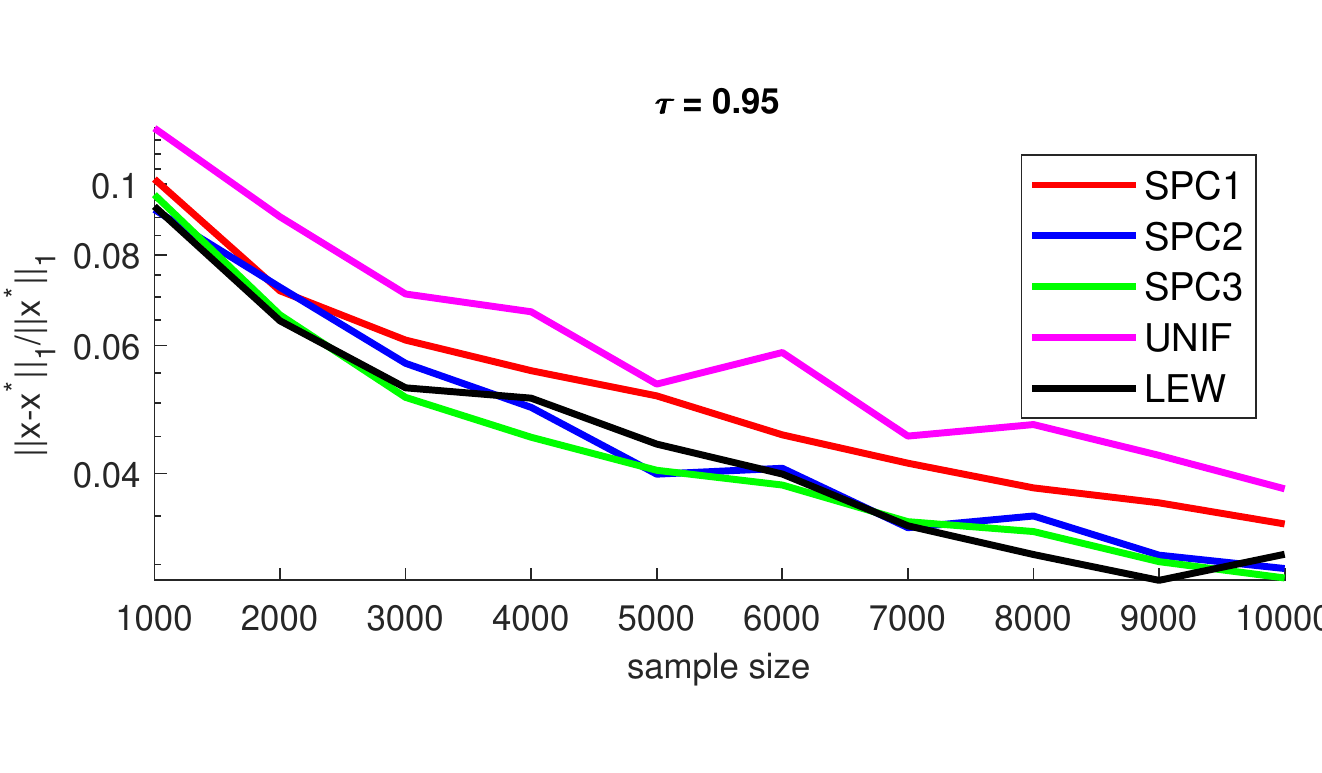}
}

\subfloat[$\tau=0.5$, $\|x\!-\!x^\ast\|_\infty/\|x^\ast\|_\infty$]{%
\includegraphics[clip, trim = 3.3cm 10.8cm 4.1cm 11.05cm, width=0.32\textwidth]{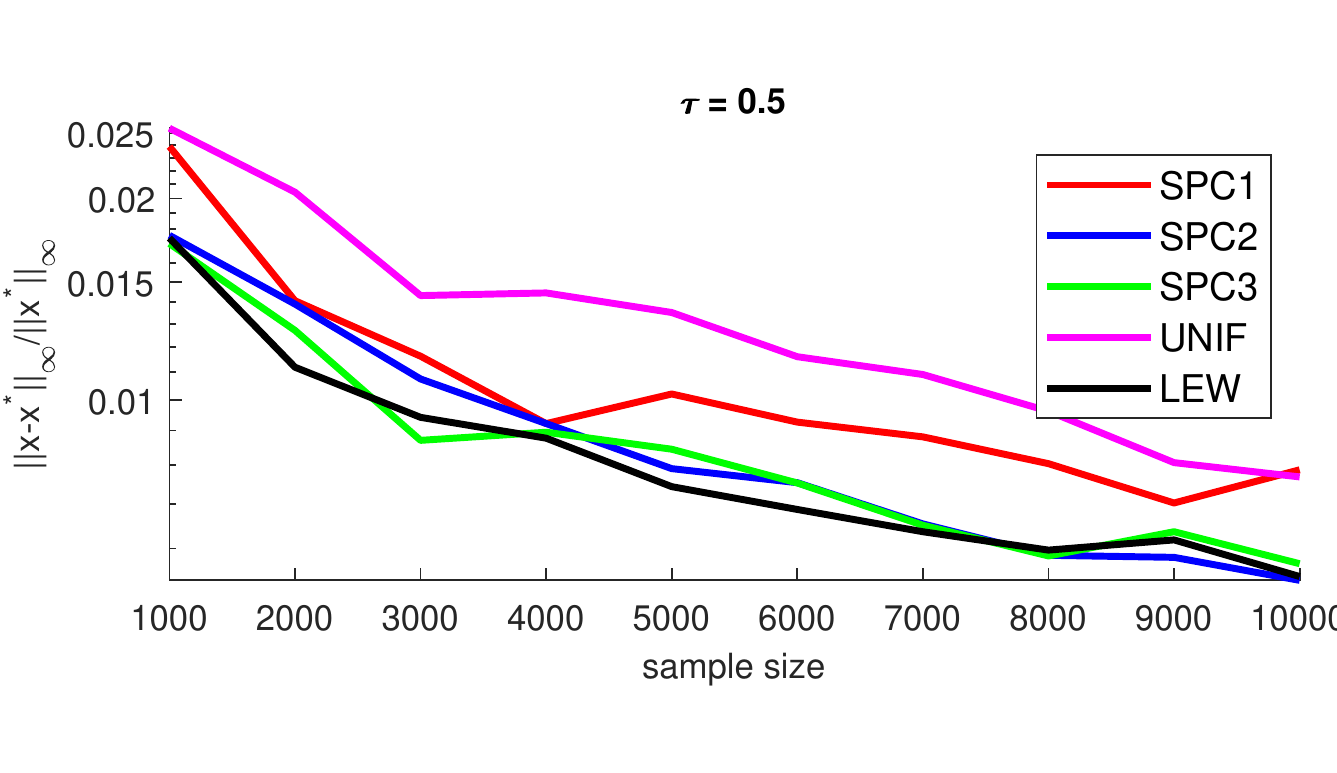}
}
\hfill
\subfloat[$\tau=0.75$, $\|x\!-\!x^\ast\|_\infty/\|x^\ast\|_\infty$]{%
\includegraphics[clip, trim = 3.3cm 10.8cm 4.1cm 11.05cm, width=0.32\textwidth]{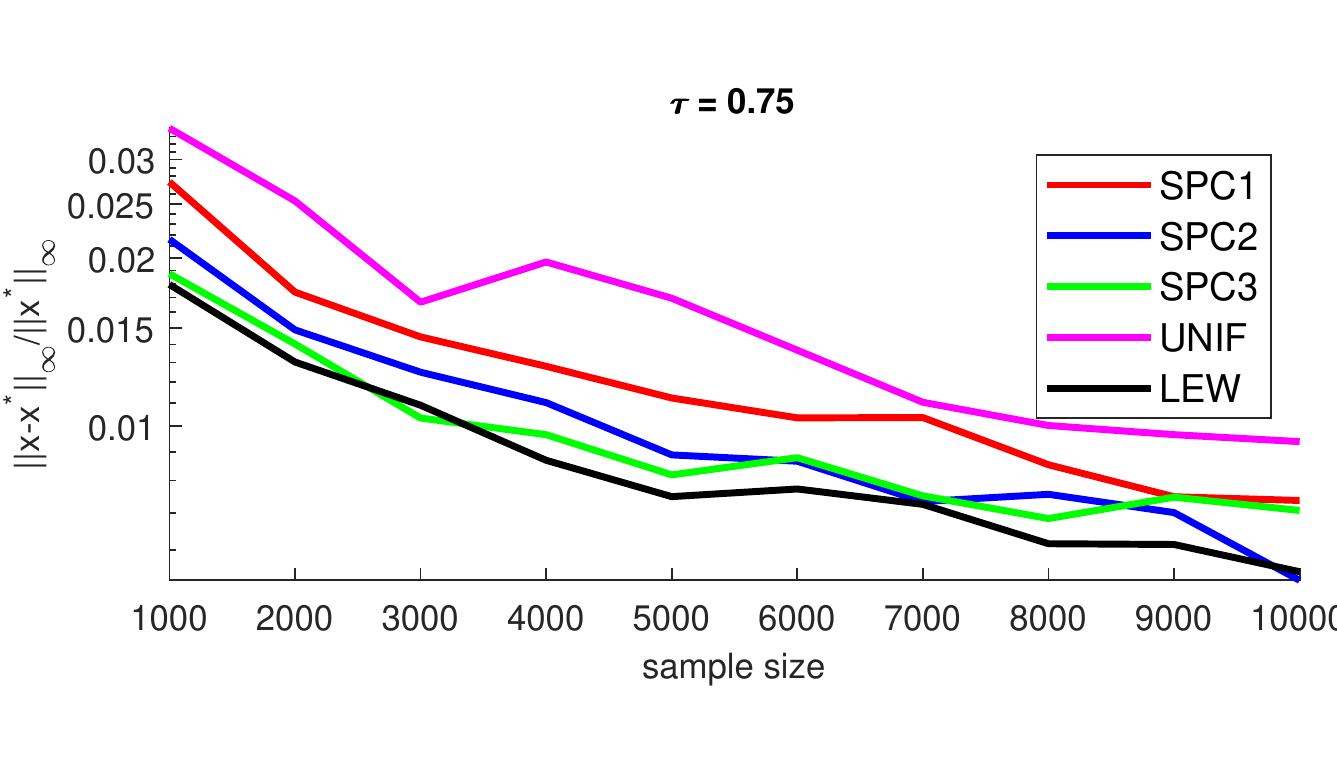}
}
\hfill
\subfloat[$\tau=0.95$, $\|x\!-\!x^\ast\|_\infty/\|x^\ast\|_\infty$]{%
\includegraphics[clip, trim = 3.45cm 10.8cm 4.1cm 11.05cm, width=0.32\textwidth]{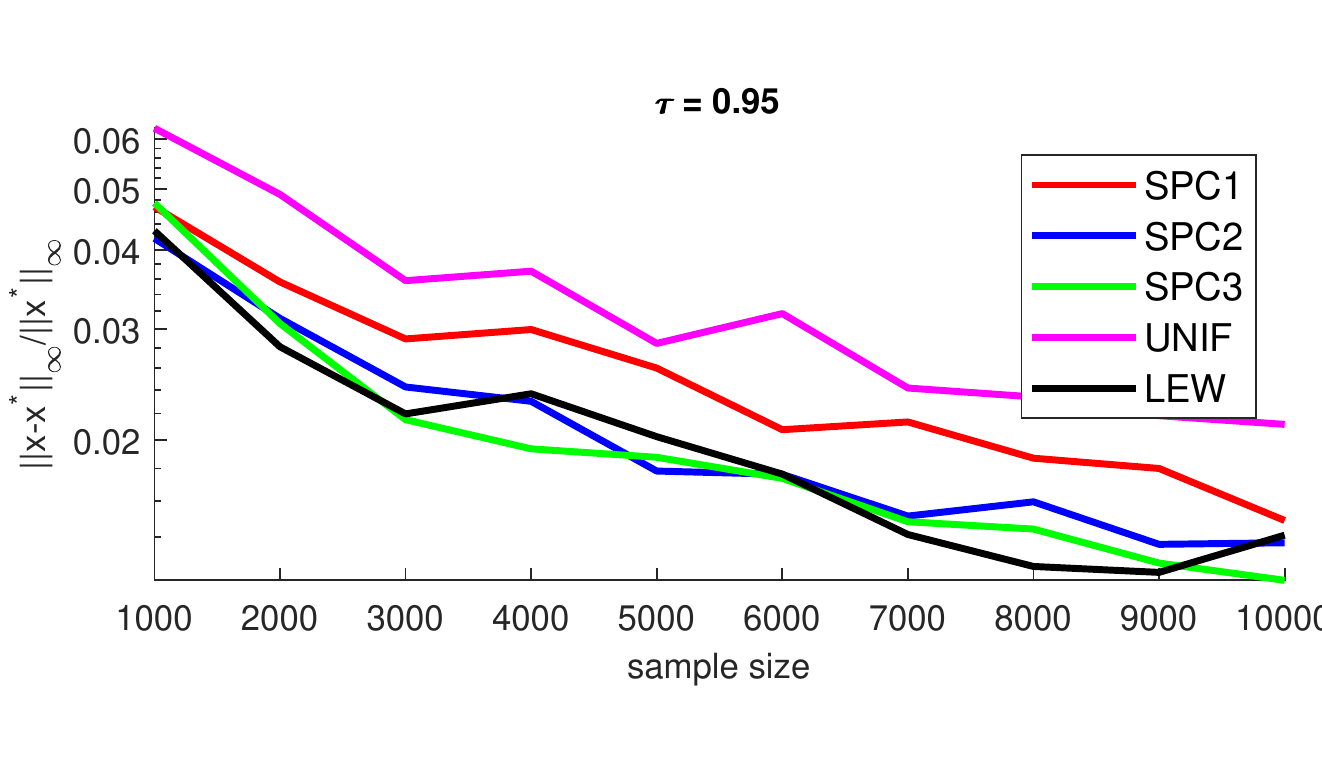}
}

\caption{\small Mean of the relative errors of the solution vector (measured in three different norms, the $\ell_2$ norm, the $\ell_1$ norm and the $\ell_\infty$ norm) of $5$ different methods on census data. The plotted lines represent the mean of the errors of $50$ independent trials.}\label{fig:vary_s_census}
\end{figure*}

\begin{figure}
\begin{minipage}[b]{0.45\linewidth}
\includegraphics[clip, trim = 3.2cm 9.3cm 3.85cm 9.6cm, width=\textwidth]{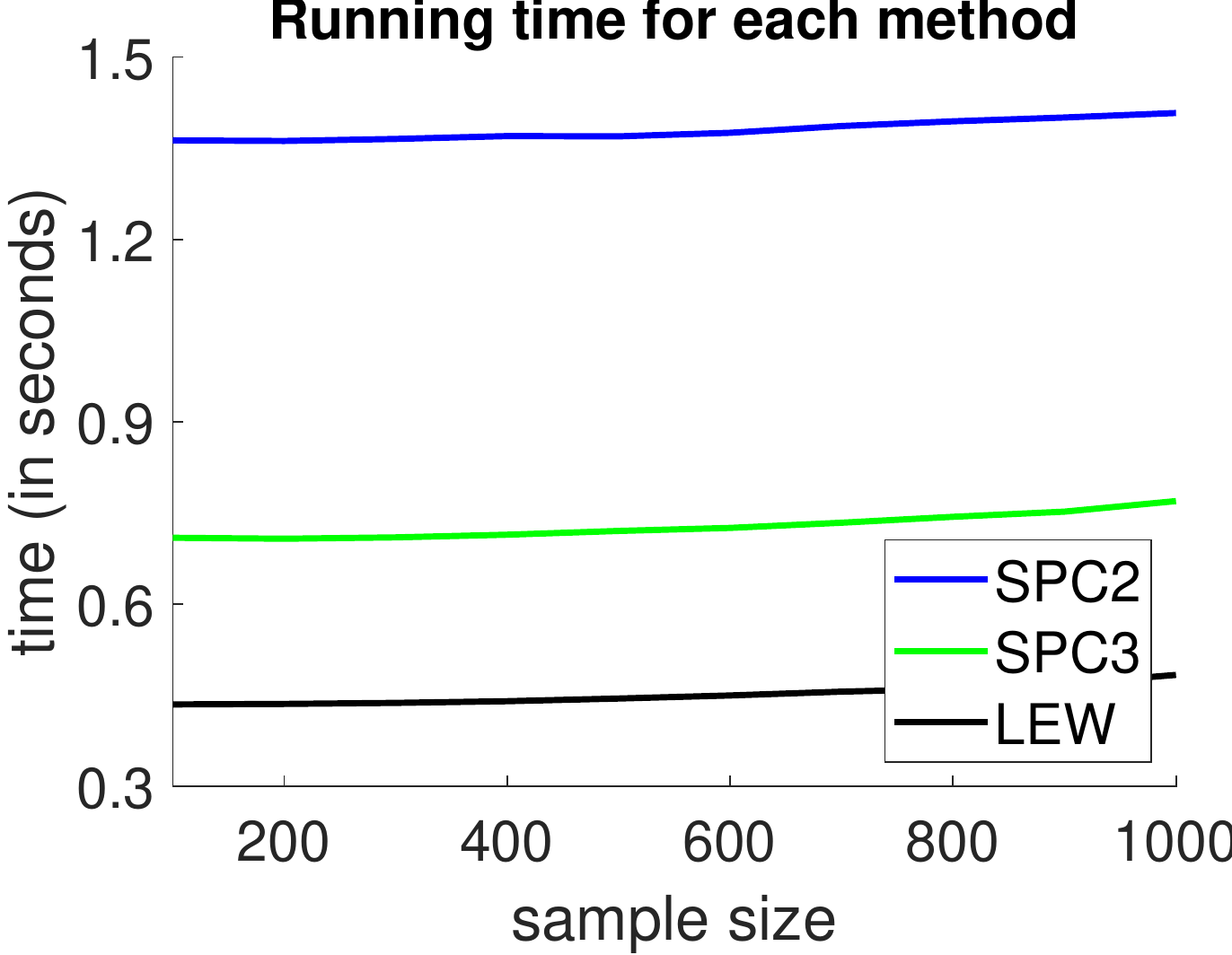}
\vspace*{-0.7cm}
\caption{\small Mean running time of $50$ independent trials on synthetic data.}\label{fig:synthetic_running time}
\end{minipage}
\hfill
\begin{minipage}[b]{0.45\linewidth}
\includegraphics[clip, trim = 3.2cm 9.3cm 3.85cm 9.6cm, width=\textwidth]{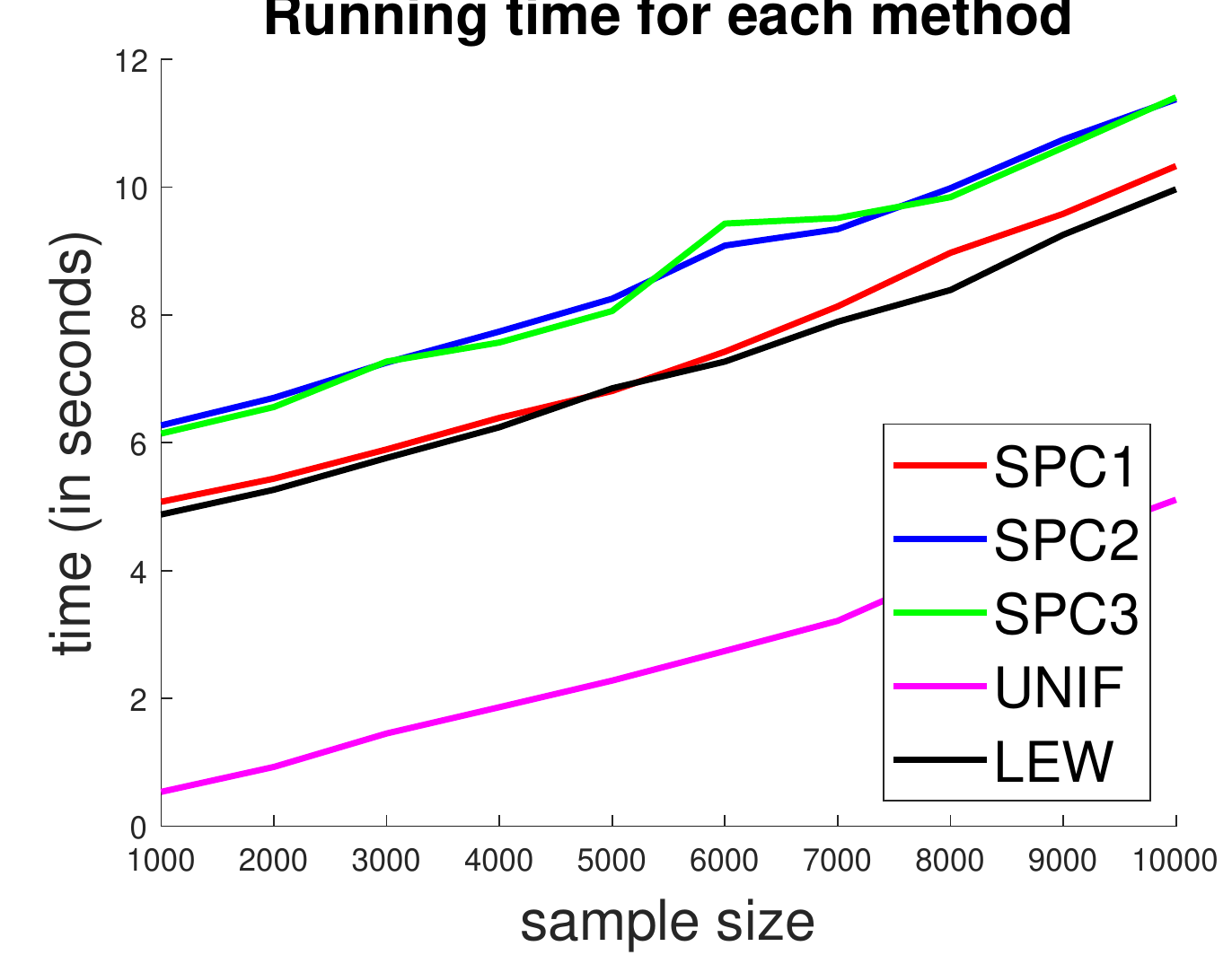}
\vspace*{-0.7cm}
\caption{\small Mean running time of $50$ independent trials on census data.}\label{fig:census_running time}
\end{minipage}
\vspace{-0.5ex}
\end{figure}

\section{Empirical Evaluation of Row Sampling Algorithms}\label{sec:exp}

In this section, we conduct an empirical evaluation\footnote{All tests are run under MATLAB 2019b on a machine of Intel Core i7-6550U CPU@2.20GHz with 2 cores.} of our sampling scheme against the three algorithms proposed in~\cite{yang2013quantile}, which are the current state of the art. The three algorithms, codenamed SPC1, SPC2 and SPC3, are all based on row sampling, but the weights are very different from Lewis weights. 
We also compare with uniform sampling as a baseline, which is widely used in practice for its simplicity.
Throughout this section, we use all samplinCorollary 3.7 algorithms only for preprocessing, and solve the reduced problem optimally using linear programming.
This allows us to remove the possible bias introduced by different optimization methods and focus on the comparison of the sampling schemes.

\paragraph{Synthetic Data.} We first evaluate our sampling scheme on a synthetic dataset using the same setup as in~\cite{yang2013quantile}.
Here, we only compare with SPC2 and SPC3,
which, as demonstrated in~\cite{yang2013quantile}, significantly outperform all other sampling schemes on this dataset.

As described in~\cite{yang2013quantile}, the data set is generated in the following manner. Each row of the matrix $A\in \RR^{n\times d}$ is a canonical basis vector. Suppose the number of occurrences of $e_j$ is $c_j$, where $c_j=qc_{j-1}$, for $j= 2, \dots, d$ and some $q\in (1,2]$ --- such distributions of rows introduce imbalanced measurements and therefore increase the difficulty of sampling. The true vector $x^\ast\in\RR^d$ has i.i.d.\@ standard normal coordinates. Let $b^\ast =A x^\ast$. The noise vector $\nu$ is generated with independent Laplacian entries, scaled such that $\|\nu\|_1/\|b^\ast\|_1= 0.2$. The response vector $b$ is given by $b_i = 500\nu_i$ with probability $0.001$ and $b_i = b_i^\ast + \nu_i$ with probability $0.999$. Here, we adhere to the same generation process so that our experimental results are comparable to those in~\cite{yang2013quantile}.

In a similar manner to~\cite{yang2013quantile}, rather than determining the sample size from a given distortion parameter $\eps$, we vary the quantile parameter $\tau\in \{ 0.5,0.75,0.95\}$, and for each $\tau$, vary the sample size among $\{100,200,\dots,1000\}$. Here the parameter $\tau$ corresponds to the parameter in the function $h_\tau(t)$ defined in Equation~\eqref{equ:quantile_def1} rather than the renormalized function $\rho_\tau(t)$ defined in Equation~\eqref{equ:quantile_def2}. We choose $n=10^5$ and $d=50$.

See Figure~\ref{fig:vary_s} for results on the synthetic dataset and Figure~\ref{fig:synthetic_running time} for the running time. When the same number of samples is at least $400$, our sampling scheme (LEW) achieves the best approximation to $x$, outperforming both SPC2 and SPC3 in all three norms ($\ell_\infty$, $\ell_1$ and $\ell_2$) under the setting $\tau = 0.5$ and $0.75$. When $\tau = 0.95$, our algorithm achieves a similar approximation error to that of SPC2, which are both better than SPC3.  This suggests that when the input is highly imbalanced and therefore more difficult to process, our method is accurate and stable with a faster running time.

\paragraph{Census Data.} We also evaluate our algorithm on a real-world dataset, the U.S.\ 2000 Census data,\footnote{http://www.census.gov/census2000/PUMS5.html} which consists of salary and related information on people who met certain criteria. As in~\cite{yang2013quantile}, we conduct an experiment on the same 5\% sample of the census data. In the corresponding regression problem, the matrix $A \in \mathbb{R}^{5048299\times 11}$ is the feature matrix and the vector $b \in \mathbb{R}^{5048299}$ records the salary of the sampled people. Our goal is to solve the regression problem $\min_x h_\tau(Ax-b)$, whose minimizer is denoted by $x^\ast$. We conduct a similar experiment as above with same values of $\tau$ and varying sample sizes in $\{1000,2000,\dots,10000\}$, which are at most 0.2\% of the data size.

See Figure~\ref{fig:vary_s_census} for approximation errors to $x^\ast$ on the census data, where we consider all methods in~\cite{yang2013quantile}. Our method, quite naturally, outperforms uniform sampling and SCP1 by a considerable margin, and  achieves almost the same relative error as, or a marginally smaller relative error in certain cases than, SPC2 and SPC3. The running time of different methods are compared in Figure~\ref{fig:census_running time}. Our method runs significantly faster than SPC2 and SPC3, suggesting that even on relatively easy-to-process real-world data, our method is robust with a superior running time.



\section{Conclusion} \label{sec:conclusion}
In this paper, we give the first row sampling algorithm for the quantile loss function with sample complexity nearly linear in the dimensionality of the data. 
Technically, we show that Lewis weights sampling can be applied in row sampling algorithms for a variety of loss functions.

There are a few interesting problems that remain open. 
Is it possible to prove lower bounds on sample complexity of row sampling algorithms for the quantile loss function using techniques in~\cite{li2020tight}? Furthermore, is it possible to use Lewis weights sampling to obtain nearly-linear sample complexity for other loss functions, e.g.\@ the $M$-estimators studied in~\cite{clarkson2015sketching, clarkson2015input}?
Answering these problems would lead to a better understanding of the power and limitation of row sampling algorithms.

\subsection*{Acknowledgements}
Y.~Li was supported in part by a Singapore Ministry of Education (AcRF) Tier 2 grant MOE2018-T2-1-013. H.~Zhang was supported by an National Science Foundation (NSF) grant IIS-1814056. The authors would like to thank Tianyi Zhang for insightful discussions.

\bibliographystyle{abbrv}

\bibliography{ref}

\appendix
\section{Additional Background}
\subsection{Rademacher Process}
We need the following classical results regarding comparison of Rademacher processes.

\begin{lemma}[{\cite[Proposition 1]{LT89}}]\label{thm:rademacher_comparison} Let $F:[0,\infty)\to [0,\infty)$ be convex and increasing. Let $\sigma_1,\dots,\sigma_n$ be a Rademacher sequence. Then for any bounded subset $T\subset \RR^n$ it holds that
\[
\E F\left(\sup_{t\in T} \left|\sum_{i=1}^n \sigma_i |t_i| \right| \right) \leq 2\E F\left(\sup_{t\in T} \left| \sum_{i=1}^n \sigma_i t_i \right|\right).
\]
\end{lemma}

\begin{lemma}[Contraction Principle, Theorem 4.4 in \cite{LT89}]\label{thm:contraction} Let $F:[0,\infty)\to [0,\infty)$ be convex. Let $\sigma_1,\dots,\sigma_n$ be a Rademacher sequence. Then for any finite sequence $(x_i)$ and any real numbers $(\alpha_i)$ such that $|\alpha_i|\leq 1$ for every $i$, it holds that
\[
\E F\left(\left|\sum_{i=1}^n \sigma_i a_i x_i \right| \right) \leq \E F\left( \left| \sum_{i=1}^n \sigma_i x_i\right|\right).
\]
\end{lemma}

The next one is a consequence of the contraction principle (see, e.g., Theorem 4.12 in~\cite{LT91}).

\begin{lemma}\label{thm:rademacher_contraction} Let $F:[0,\infty)\to [0,\infty)$ be convex and increasing. Suppose that $f(x),g(x)$ are nonnegative functions such that $|f(x)-f(y)|\leq L|g(x)-g(y)|$ for all $x,y\in\RR$. Let $\sigma_1,\dots,\sigma_n$ be a Rademacher sequence. Then for any bounded subset $T\subset \RR^n$ it holds that
\[
\E F\left(\frac{1}{2L}\sup_{t\in T} \left|\sum_{i=1}^n \sigma_i f(t_i) \right| \right) \leq \E F\left(\sup_{t\in T} \left| \sum_{i=1}^n \sigma_i g(t_i) \right|\right).
\]
\end{lemma}

\subsection{Convex Optimization}
We recall some results in convex optimization.
\begin{definition}\label{def:lip}
A function $f : \mathbb{R}^d \to \mathbb{R}$ is $G$-Lipschitz if for any $x, y \in \mathbb{R}^d$,
\[
|f(x) - f(y)| \le G \|x - y\|_2.
\]
\end{definition}
We need the following result in~\cite{allen2017katyusha}.
\begin{theorem}[{\cite[Corollary 3.7]{allen2017katyusha}}]\label{thm:katyusha}
For a given function $f(x) = \frac{1}{n} \sum_{i = 1}^n f_i(x)$, let $x^* = \argmin_{x \in \mathbb{R}^d} f(x)$.
If each $f_i : \mathbb{R}^d \to \mathbb{R}$ is convex and $\sqrt{G}$-Lipschitz, then there is an algorithm that receives an initial solution $x_0 \in \mathbb{R}^d$, and outputs a solution $x \in \mathbb{R}^d$ satisfying $\E[f(x)] - f(x^*)$ in 
\[
T = O\left(n \log \frac{f(x_0) - f(x^*)}{\varepsilon}  + \frac{\sqrt{nG} \|x_0 - x^*\|}{\varepsilon} \right)
\]
stochastic subgradient iterations.

\end{theorem}

\section{Missing Proofs}
\label{sec:proof}

\subsection{Proof of Theorem~\ref{thm:calc_graph}}

\begin{proof}
By Lemma~2.4 in~\cite{cohen2015p}, to calculate approximate $\ell_1$ Lewis weights of rows of a matrix $A$, it suffices to calculate approximate leverage scores of rows of matrices of the form $WA$, for $O(\log n)$ different diagonal matrices $W \in \mathbb{R}^n$.

When $A$ is the edge-vertex incidence matrix of a graph $G$, $WA$ is the edge-vertex incidence matrix of a reweighted graph $G'$. 
In this case, leverage scores of rows of $WA$ are the effective resistances of $G'$ (cf.~\cite{drineas2010effective}), which can be computed in $\widetilde{O}(m)$ time using the algorithm in~\cite{spielman2011graph}.
\end{proof}

\subsection{Proof of Theorem~\ref{thm:quantile-sample}}\label{sec:claim_quantile_sampling}

\begin{proof}[Proof of Theorem~\ref{thm:quantile-sample}]
Suppose that $i_1,i_2,\dots,i_N$ are the indices of $\{e_1,\dots,e_m\}$ randomly chosen in the construction of $S$. Since $\phi(\alpha t) = \alpha\phi(t)$ for $\alpha > 0$, for each coordinate of $\wt{A}x$, we have
\begin{align*}
\E\phi((\wt{A}x)_k) = \E\phi\left(\frac{1}{p_{i_k}}\langle A_{i_k}, x\rangle\right) &= \sum_{j=1}^n \frac{1}{p_j}\phi(A_j^\top x)\cdot \frac{p_j}{N} = \frac{1}{N}\phi(Ax)
\end{align*}
and
\[
\E\phi(\wt{A}x) = \sum_{k=1}^N \E\phi((\wt{A}x)_k) = \phi(Ax).
\]
We just assume that $p_i\geq C_s\epsilon^{-2}w_i\log N$ for now and shall rescale $\eps$ in the end. The $\ell_1$ Lewis weight sampling result in~\cite{cohen2015p} implies that with probability at least $1-1/\poly(N)$ (over $i_1,\dots,i_N$),
\begin{equation}\label{eqn:condition}
\|\wt{A}x\|_1\leq C_1\|Ax\|_1,\quad \forall x\in \RR^d.
\end{equation}
We shall condition on this event in the rest of the proof.

Our goal is to derive a tail inequality for
\[
\sup_{x: \phi(Ax) = 1} \left| \phi(\wt{A}x) - 1 \right| = \sup_{x: \phi(Ax) = 1} \left| \sum_{k=1}^N \frac{\phi(\langle A_{i_k}, x\rangle)}{p_{i_k}} - 1\right|.
\]
We shall look at a higher moment of the deviation and apply Markov's inequality. To this end, we investigate
\begin{align*}
M &= \E_{S} \left(\sup_{\phi(Ax) = 1} \big| \phi(\wt{A}x) - 1\big|\right)^\ell = \E_{i}\left(\sup_{\phi(Ax) = 1} \left|\sum_{k=1}^N	\frac{\phi(\langle A_{i_k}, x\rangle)}{p_{i_k}} - 1\right|\right)^\ell,
\end{align*}
where $i=(i_1,i_2,\dots,i_N)$.

A standard symmetrization argument gives that
\begin{align}
\label{eqn:m}
M\le 2^{\ell} \E_{i,\sigma} \left[\left(\sup_{\phi(Ax)= 1}\left|\sum_{k=1}^N	\sigma_k\frac{\phi(\langle A_{i_k}, x\rangle)}{p_{i_k}}\right|\right)^\ell\right],
\end{align}
where $\sigma=(\sigma_1,\sigma_2,\ldots, \sigma_N)$ is a Rademacher sequence. It follows that
\begin{equation}
\label{eqn:quantile_aux0}
\begin{aligned}
 &\quad\ \E_{i,\sigma} \left(\sup_{\phi(Ax) = 1}\left|\sum_{k=1}^N	\sigma_k\frac{\phi(\langle A_{i_k}, x\rangle)}{p_{i_k}}\right|\right)^\ell \\
 &= \E_{i,\sigma} \sup_{\phi(Ax) = 1} \left|a\sum_{k=1}^N\sigma_k\frac{|\langle A_{i_k}, x\rangle|}{p_{i_k}} + b\sum_{k=1}^N\sigma_k\frac{\langle A_{i_k}, x\rangle}{p_{i_k}}\right|^\ell\\
 &\le \E_{i,\sigma} \sup_{\phi(Ax) = 1}\left(\left|a\sum_{k=1}^N\sigma_k\frac{|\langle A_{i_k}, x\rangle|}{p_{i_k}}\right| + b\left|\sum_{k=1}^N\sigma_k\frac{\langle A_{i_k}, x\rangle}{p_{i_k}}\right|\right)^\ell\\
 &\le 2^{\ell-1}\cdot \E_{i,\sigma} \sup_{\phi(Ax) = 1} \left[ \left|a\sum_{k=1}^N\sigma_k\frac{|\langle A_{i_k}, x\rangle|}{p_{i_k}}\right|^\ell + \left(b\left|\sum_{k=1}^N\sigma_k\frac{\langle A_{i_k}, x\rangle}{p_{i_k}}\right|\right)^\ell\right]\\
 &\le (2a^\ell+b^\ell)\cdot 2^{\ell-1} \cdot \E_{i,\sigma} \sup_{\phi(Ax) = 1} \left|\sum_{k=1}^N\sigma_k\frac{\langle A_{i_k}, x\rangle}{p_{i_k}}\right|^\ell\\
 &\leq 3a^\ell 2^{\ell-1} \E_{i,\sigma} \sup_{\phi(Ax) = 1} \left|\sum_{k=1}^N\sigma_k\frac{\langle A_{i_k}, x\rangle}{p_{i_k}}\right|^\ell,
\end{aligned}
\end{equation}
where the penultimate inequality follows from Lemma~\ref{thm:rademacher_comparison}.

For now, condition on the choices of $i_1,i_2,\ldots, i_N$. By Lemma B.1 in~\cite{cohen2015p},  there exists a matrix $A'$ of $O(d^2)$ rows such that
\begin{equation}\label{eqn:append}
\|A'x\|_1\leq C_1\|Ax\|_1,\quad \forall x\in \RR^d,
\end{equation}
and the Lewis weights of $A'$ are $O(1/d)$. Append $A'$ to the matrix $\wt{A}$, obtaining a new matrix $A''$ of $N' = N + O(d^2)$ rows, and it is a direct consequence of the contraction principle (Lemma~\ref{thm:contraction}) that
\begin{equation}\label{eqn:quantile_aux1}
\E_{\sigma} \sup_{\phi(Ax) = 1} \left|\sum_{k=1}^N\sigma_k\frac{\langle A_{i_k}, x\rangle}{p_{i_k}}\right|^\ell \leq \E_\sigma \sup_{\phi(Ax) = 1} \left|\sum_{k=1}^{N'} \sigma_k \langle A_k'', x \rangle \right|^\ell,
\end{equation}
since the sum on the right-hand side has more terms.

Furthermore, it can be proved that the Lewis weights of $A''$ are all $O(\eps^2/\log N')$ (see~\cite{cohen2015p}). In this case, for an appropriate choice of $N$ (and thus $N'$), it is implicitly shown in the proof of Theorem 2.3 in~\cite{cohen2015p}  that
\begin{equation}\label{eqn:quantile_aux2}
\E_{\sigma} \sup_{\|A''x\|_{1} =  1}\left|\sum_{k=1}^{N'} \sigma_k \langle A_k'', x\rangle \right|^\ell \le \frac{\eps^\ell}{\poly(N)}
\end{equation}
for some $\ell = \Theta(\log N')$. The next step is to relate the right-hand side of \eqref{eqn:quantile_aux1} to the left-hand side of \eqref{eqn:quantile_aux2}. Note that
\begin{align*}
\|A''x\|_1 &= \|\wt{A}x\|_1 + \|A'x\|_1 \\
&\leq C_1\|Ax\|_1 + C_2\|Ax\|_1\quad  (\text{by \eqref{eqn:condition} and \eqref{eqn:append}})\\
&\leq (C_1+C_2)B\phi(Ax)
\end{align*}
and thus
\begin{align*}
\sup_{\phi(Ax) =  1}\left|\sum_{k=1}^{N'} \sigma_k \langle A_k'', x\rangle \right|^\ell \leq  (C_1+C_2)^\ell B^\ell \sup_{\|A''x\|_{1} =  1}\left|\sum_{k=1}^{N'} \sigma_k \langle A_k'', x\rangle \right|^\ell.
\end{align*}
Taking expectation over $\sigma$ on both sides, we obtain using~\eqref{eqn:quantile_aux2} that
\[
\E_\sigma \sup_{\|Ax\|_{1} =  1}\left|\sum_{k=1}^{N'} \sigma_k \langle A_k'', x\rangle \right|^\ell \leq ((C_1+C_2)B)^\ell \frac{\epsilon^\ell}{\poly(N)}.
\]
Taking expectation over $i_1,\dots,i_N$ on both sides subject to the conditioning~\eqref{eqn:condition} and combining with \eqref{eqn:m}, \eqref{eqn:quantile_aux0} and \eqref{eqn:quantile_aux1}, we obtain that 
\[
M \leq C_3(C_4aB)^\ell \frac{\epsilon^\ell}{\poly(N)}
\]
The result follows from Markov's inequality with a rescaling of $\eps$ by a factor of $1/(C_5aB)$.
\end{proof}

\subsection{Proof of Theorem~\ref{thm:main_quantile}}\label{proof:quantile}
\begin{proof}
Recall that $\rho_{\tau}(x) = -\tau x$ if $x\le 0$ and $\rho_{\tau}(x)= x$ if $x\ge 0$. We can rewrite $\rho_{\tau}(x) = (1+\tau)|x|/2+ (1-\tau) x/2$. Also $\rho_\tau(x) \leq \|x\|_1/\tau$, and thus we can take $B=1/\tau$ in Theorem~\ref{thm:quantile-sample}.
By Theorem~\ref{thm:calc}, we can obtain $\{\overline{w}_i\}_{i=1}^n$, such that with probability $1 - 1 / \poly(d)$, for all $i \in [n]$, $w_i \le \overline{w}_i  \le 2 w_i$, where $\{w_i\}_{i=1}^n$ are the $\ell_1$ Lewis weights of $A$.
Now we invoke the row sampling algorithm in Theorem~\ref{thm:quantile-sample} and the fact in Lemma~\ref{lem:sum_lewis} to finish the proof.
\end{proof}

\subsection{Proof of Lemma~\ref{lem:sketch}}
\begin{proof}
By Theorem~\ref{thm:main_quantile}, with probability at least $0.9$, for all $x \in \mathbb{R}^d$, 
\[
 \left(1 - \frac{\varepsilon}{9}\right)\rho_\tau(Ax - b)  \le \rho_\tau(\wt{A}x - \wt{b}) \le \left(1 + \frac{\varepsilon}{9}\right) \rho_\tau(Ax - b).
\]
Let $x^{\mathrm{opt}} = \argmin_{x \in \mathbb{R}^d} \rho_\tau(Ax - b)$, we have
\begin{align*}
\rho_\tau(Ax^* - b) &= \rho_\tau(AR^{-1}\overline{x} - b) \\
&\le \frac{1}{1 - \varepsilon / 9}\rho_\tau(\wt{A}R^{-1}\overline{x} - \wt{b})  \\
&\le  \frac{1 + \varepsilon / 3}{1 - \varepsilon / 9}\rho_\tau(\wt{A}R^{-1} (Rx^{\mathrm{opt}}) - \wt{b}) \tag{By Equation~\eqref{equ:approx_sketched}}\\
&=  \frac{1 + \varepsilon / 3}{1 - \varepsilon / 9}\rho_\tau(\wt{A}x^{\mathrm{opt}} - \wt{b}) \\
&\le  \frac{(1 + \varepsilon / 3)(1 + \varepsilon / 9)}{1 - \varepsilon / 9}\rho_\tau(Ax^{\mathrm{opt}} - b)\\
& \le (1 + \varepsilon) \rho_\tau(Ax^{\mathrm{opt}} - b).\qedhere
\end{align*}
\end{proof}

\subsection{Proof of Lemma~\ref{lem:conditions}}
We need the following claim in our proof.
\begin{claim}\label{claim:matrix_quantile}
For a matrix $A \in \mathbb{R}^{n \times d}$ and a vector $y \in \mathbb{R}^n$,
\[
\|A^{\top} y\|_2 \le \frac{1}{\tau} \rho_\tau(y) \max_{1 \le i \le N} \|A_i\|_2.
\]
\end{claim}
\begin{proof}
Observe that $\sum_{i = 1}^n |y_i| \le \frac{1}{\tau}\rho_\tau(y) $. It follows that
\begin{align*}
\|A^{\top} y\|_2 = \left\|\sum_{i = 1}^N y_i A_i\right\|_2 &\le \sum_{i = 1}^N |y_i| \|A_i\|_2 \le \max_{1 \le i \le N} \|A_i\|_2 \cdot \sum_{i = 1}^N |y_i| \le \frac{1}{\tau} \rho_\tau(y) \max_{1 \le i \le N} \|A_i\|_2.\qedhere
\end{align*}
\end{proof}

Now we are ready to prove Lemma~\ref{lem:conditions}.
\begin{proof}
By Lemma~30 in~\cite{durfee2018ell}, with probability at least $0.9$, the leverage score of each row of $\sAb$ satisfies $\tau_i(\sAb) = O(d / N)$.
We condition on this event in the remaining part of the proof.
By Lemma 2 in~\cite{cohen2015uniform}, 
\[
\tau_i(\sAb) = \min_{\sAb^{\top}x = (\sAb)_i} \|x\|_2^2,
\]
and
\[
\tau_i(\wt{A}) = \min_{\wt{A}^{\top}x = \wt{A}_i} \|x\|_2^2,
\]
which implies $\tau_i(\wt{A}) \le \tau_i(\sAb)$.
Thus, $\tau_i(\wt{A}) = O(d / N)$.
Since $\wt{A}R^{-1} = Q$ has orthonormal columns, for each row $Q_i$ of $Q$,
\[
\|Q_i\|_2^2 = \tau_i(\wt{A}) = O(d / N).
\]
Here we used the standard fact of the relation between leverage scores and the QR decomposition. 
%
See Definition 2.6 in~\cite{woodruff2014sketching} for details.

Let $f(x) = \rho_\tau(\wt{A}R^{-1}x - \wt{b})$, we can write
\[
f(x) = \sum_{i = 1}^n \rho_{\tau}(\langle (\wt{A}R^{-1})_i, x \rangle - b_i) = \frac{1}{N}\sum_{i=1}^n f_i(x),
\]
where $f_i(x) = N\cdot \rho_{\tau}(\langle (\wt{A}R^{-1})_i, x \rangle - b_i)$. Let $g_i(x) = \langle (\wt{A}R^{-1})_i, x \rangle - b_i$, then
\[
\|\nabla g_i(x)\|_2 \le \|(\wt{A}R^{-1})_i\|_2 = \|Q_i\|_2 = O(\sqrt{d/N}),
\]
which implies that each $g_i(x)$ is $O(\sqrt{d/N})$-Lipschitz. Observe that $\rho_\tau(\cdot)$ is $1$-Lipschitz, it follows that $f_i(x) = N\rho_\tau(g_i(x))$ is $O(\sqrt{Nd})$-Lipschitz.

By Claim~\ref{claim:matrix_quantile}, we have
\begin{align*}
 \|x_0 - \wt{x^{\mathrm{opt}}} \|_2 &=  \|(\wt{A}R^{-1})^{\top} b - \wt{x^{\mathrm{opt}}} \|_2\\
&=  \|(\wt{A}R^{-1})^{\top} b - \wt{x^{\mathrm{opt}}} \|_2 \\
&=  \|(\wt{A}R^{-1})^{\top} (b - (\wt{A}R^{-1}) \wt{x^{\mathrm{opt}}}) \|_2 \\
&\le \sqrt{\frac{d}{N\tau^2}} \cdot \rho_\tau((\wt{A}R^{-1}) \wt{x^{\mathrm{opt}}} - b).
\end{align*}
Finally, for any vector $y \in \mathbb{R}^N$, we have
\[
\tau \|y\|_2 \le \tau \|y\|_1 \le \rho_\tau(y) \le \|y\|_1 \le \sqrt{N} \|y\|_2.
\]

Since $\wt{A} R^{-1} = Q$ has orthonormal columns, $x_0 = (\wt{A}R^{-1})^{\top} b$ is the optimal solution to $\min_{x \in \mathbb{R}^d} \|\wt{A}R^{-1} x_0 - \wt{b}\|_2$.
Thus, we have $\|\wt{A}R^{-1} x_0 - \wt{b}\|_2 \le \|\wt{A}R^{-1} \wt{x^{\mathrm{opt}}}  - \wt{b}\|_2$.
Therefore, 
\begin{align*}
\rho_\tau(\wt{A}R^{-1} x_0 - \wt{b}) &\le \sqrt{N}\|\wt{A}R^{-1} x_0 - \wt{b}\|_2 \\
&\le \sqrt{N}\|\wt{A}R^{-1} \wt{x^{\mathrm{opt}}}  - \wt{b}\|_2 \\
&\le \sqrt{\frac{N}{\tau^2}}\rho_\tau(\wt{A}R^{-1} \wt{x^{\mathrm{opt}}}  - \wt{b}).
\end{align*}
\end{proof}

\subsection{Proof of Lemma~\ref{lem:katyusha}}
\begin{proof}
The initial solution $x_0 = (\wt{A}R^{-1})^{\top} b$ can be calculated in $O(Nd)$ time.
We condition on the event in Lemma~\ref{lem:conditions}.
By Theorem~\ref{thm:katyusha}, 
with probability at least $0.9$, after 
\begin{align*}
T = O\left(N \log \frac{f(x_0) - f(\wt{x^{\mathrm{opt}}} )}{\varepsilon \cdot f(\wt{x^{\mathrm{opt}}} )} + \frac{\sqrt{N^2 d} \|x_0 - x^*\|_2}{\varepsilon \cdot f(\wt{x^{\mathrm{opt}}} )}\right) = O\left(N \log \frac{N}{\varepsilon \tau^2} + \frac{d N^{1/2}}{\varepsilon \tau}\right) = \wt{O}\left(\frac{d^{1.5}}{\tau^2 \varepsilon^2} \right)
\end{align*}
stochastic subgradient iterations, we can find a solution $\overline{x}$ such that
\begin{align*}
\E\left[\rho_\tau(\wt{A}R^{-1}\overline{x} - \wt{b}) -\min_{x \in \mathbb{R}^d} \rho_\tau(\wt{A}R^{-1}x - \wt{b})\right] \le \varepsilon / 30 \cdot  \rho_\tau(\wt{A}R^{-1}x - \wt{b}).
\end{align*}
By Markov's inequality, with probability at least $0.9$, we have
\[
\rho_\tau(\wt{A}R^{-1}\overline{x} - \wt{b}) \le  (1 + \varepsilon / 3) \cdot \min_{x \in \mathbb{R}^d} \rho_\tau(\wt{A}R^{-1}x - \wt{b}).
\]
Furthermore, each stochastic subgradient can be calculated in $O(d)$ time, since
\[
\nabla f_i(x)= \begin{cases}
\operatorname{sign}(\langle A_i, x \rangle - b_i) \cdot A_i & \text{if }\langle A_i, x \rangle - b_i \geq 0 \\
\tau \cdot \operatorname{sign}(\langle A_i, x \rangle - b_i) \cdot A_i & \text{otherwise}
\end{cases}.
\]
where we choose a subgradient $\nabla f_i(x) = 0$ at the nondifferentiable points $x$.
\end{proof}

\subsection{Proof of Theorem~\ref{thm:reg}}
\begin{proof}
Finding the QR decomposition of the concatenated matrix $\sAb$ can be done in $\wt{O}(d^{\omega}/(\varepsilon^2 \tau^2))$ time (see Lemma 33 in~\cite{durfee2018ell}).
By Lemma~\ref{lem:katyusha}, we can obtain $\overline{x}$ such that $\rho_\tau(\wt{A}R^{-1}\overline{x} - \wt{b}) \le (1 + \varepsilon / 3) \cdot \min_{x \in \mathbb{R}^d} \rho_\tau(\wt{A}R^{-1}x - \wt{b})$  in $\wt{O}(d^{2.5}/(\tau^2 \varepsilon^2))$ time and succeeds with probability at least $0.8$.
By Lemma~\ref{lem:sketch}, with probability at least $0.9$, the obtained solution $x^* \in \mathbb{R}^d$ satisfies $\rho_\tau(Ax^*-b) \le (1 + \varepsilon) \min_{x \in \mathbb{R}^d}\rho_\tau(Ax-b)$.
We complete the proof of the theorem by taking a union bound over the two events mentioned above.
\end{proof}
\subsection{Proof of Lemma~\ref{lem:balanced_graph}}
\begin{proof}
By permuting the columns we may assume without loss of generality that $x_1\leq x_2\leq \cdots\leq x_n$. Suppose that $a = x_1 < x_n = b$, otherwise $Ax=0$ and the desired inequality holds automatically. 

Let $I = \{i\in [n]: (Ax)_i\geq 0\}$. 
Let $u = \sum_{i\in I} A_i$ and $v = \sum_{i\in I}A_i - \sum_{i\not\in I} A_i$, then $\rho_0(Ax) = \langle u,x\rangle$ and $\rho_1(Ax) = \langle v,x\rangle$ for all $x\in P$, where $P = \{x\in [a,b]^n: a = x_1\leq x_2\leq \cdots \leq x_n = b\}$ is a polytope.

Observe that $\rho_1(Ax)\neq 0$ on $P$ and thus the function $f(x) = \langle u,x\rangle/\langle v,x\rangle$ attains its minimum value $\lambda$ on the compact set $P$. Suppose that $\langle u,x^*\rangle =\lambda \langle v,x^*\rangle$ for some $x^*\in P$. We claim that $x^*$ is also the minimizer of $\langle u-\lambda v,x\rangle$ on $P$. Indeed, if $\langle u-\lambda v,x'\rangle < \langle u-\lambda v,x^*\rangle$ for some $x'\in P$, then $\langle u,x'\rangle/\langle v,x'\rangle < \lambda$, contradicting the minimality of $x'$.

Now, since $x^*$ is a minimizer of $\langle u-\lambda v,x\rangle$ on the polytope $P$, it must be some vertex of $P$, that is, there exists $k$ such that $x^*_1 = \cdots = x^*_k = a$ and $x^*_{k+1} = \cdots = x^*_n = b$. Let $S = \{x_1,\dots,x_k\}$, then 
\[
\frac{\rho_0(Ax^*)}{\rho_1(Ax^*)} = \frac{w(S, V \setminus S)}{w(S, V \setminus S) + w(V \setminus S, S)} \ge \frac{1}{\alpha + 1},
\]
where the last step follows from the definition of the $\alpha$-balanced graph.
We complete the proof by noticing that \[\frac{\rho_0(Ax)}{\rho_1(Ax)} \ge \frac{\rho_0(Ax^*)}{\rho_1(Ax^*)} \ge \frac{1}{\alpha + 1}.\]
\end{proof}

\subsection{Proof of Corollary~\ref{cor:spar}}
\begin{proof}
Observe that $\rho_0(x) = \frac12 |x| + \frac12 x$. 
Moreover, by Lemma~\ref{lem:balanced_graph}, we have $\|Bx\|_1 \leq (1+\alpha)\rho_0(Bx)$ for all $x\in \RR^n$. Now we invoke Theorem~\ref{thm:quantile-sample} with $a = b = \frac12$ and $B = \alpha + 1$, which states that with probability at least $1 - 1 / \poly(n)$, for all $x \in \mathbb{R}^n$,
$(1-\eps) \rho_0(Bx) \leq \rho_0(B'x)  \leq (1+\eps) \rho_0(Bx)$.
Moreover, the rows of $B'$ are reweighted rows of $B$, which implies $B'$ is the edge-vertex matrix of a graph $G'$, whose edges are reweighted edges of $G$.
The running time of Algorithm~\ref{fig:spar} directly follows from Theorem~\ref{thm:calc_graph}.
\end{proof}

\section{Proof of Theorem~\ref{thm:generalized-sample}}\label{sec:extra_proof}

Similar to the proof of Theorem~\ref{thm:quantile-sample} and below we shall only highlight the changes. Instead of the comparison result of Lemma~\ref{thm:rademacher_comparison}, we shall use Lemma~\ref{thm:rademacher_contraction}.

Note that we have here that
\[
\E \phi^{w}(\widetilde{A}x) = \phi(Ax)
\]
and want to upper bound
\[
M := \E_S \left[\sup_{x\neq 0} \left|\frac{\phi^w(\widetilde{A}x)}{\phi(Ax)} -1\right|^\ell \right].
\]

Again by a standard symmetrization argument,
\begin{align*}
M &= \E_{S}\left[\sup_{x\neq 0} \left|\sum_{i=1}^N \frac{\phi((\wt{A}x)_i)}{p_{i_k}\phi(Ax)}  - 1\right|^\ell\right] \le 2^{\ell} \E_{i,\sigma} \left[\sup_{x\neq 0}\left|\sum_{k=1}^N	\sigma_k\frac{\phi(\langle A_{i_k}, x\rangle)}{p_{i_k}\phi(Ax)}\right|^\ell\right],
\end{align*}

where $\sigma_1,\sigma_2,\dots$ is a Rademacher sequence and $i_1,i_2,\dots$ are the indices of the rows chosen randomly during the construction of $S$.

Invoking Lemma~\ref{thm:rademacher_contraction} we have for fixed $i_1,i_2,\dots$ that
\begin{align*}
&\quad\ \E_{\sigma} \left[\sup_{x\neq 0}\left|\sum_{k=1}^N \sigma_k\frac{\phi(\langle A_{i_k}, x\rangle)}{p_{i_k}\phi(Ax)}\right|^\ell\right] \\
&\leq (2L)^\ell \E_{\sigma} \left[\sup_{x\neq 0}\left|\sum_{k=1}^N	\sigma_k\frac{ |\langle A_{i_k}, x\rangle|^p}{p_{i_k}\phi(Ax)}\right|^\ell\right] \\
&\leq \left(\frac{2L}{\gamma}\right)^\ell \E_{\sigma} \left[\left(\sup_{x\neq 0}\left|\sum_{k=1}^N \sigma_k\frac{ |\langle A_{i_k}, x\rangle|^p}{p_{i_k}\|Ax\|_p^p}\right|\right)^\ell\right] \\
&\leq \left(\frac{2L}{\gamma}\right)^\ell \E_{\sigma} \left[\left(\sup_{\|Ax\|_{p} = 1}\left|\sum_{k=1}^N	\sigma_k\frac{ |\langle A_{i_k}, x\rangle|^p}{p_{i_k}}\right|\right)^\ell\right]
\end{align*}
The problem is thus reduced to the analysis of $\ell_p$ Lewis weight sampling and it has been proved implicitly in~\cite{cohen2015p} that
\[
\E_{i,\sigma} \left[\left(\sup_{\|Ax\|_{p} = 1}\left|\sum_{k=1}^N	\sigma_k\frac{ |\langle A_{i_k}, x\rangle|^p}{p_{i_k}}\right|\right)^\ell\right] \leq \frac{(C_1\eps)^\ell}{\poly(N)}.
\]
and hence
\[
M \leq 2^{2\ell} \cdot \left(\frac{2L}{\gamma}\right)^\ell \cdot \frac{(C_1 \epsilon)^\ell}{\poly(N)} \leq \frac{(C_2L/\gamma \cdot \eps)^\ell}{\poly(N)}.
\]
The result follows from Markov's inequality with a rescaling of $\eps$ by a factor of $\gamma/(C_2 L)$.

\end{document}